\documentclass[a4paper,11pt]{article}

\input{setup.sty}

\usepackage[url=false,giveninits=true,maxbibnames=99,style=alphabetic,maxalphanames=5]{biblatex}

\title{On the complexity of estimating ground state entanglement and free energy} 
\author{Sevag Gharibian\footnote{Department of Computer Science and Institute for Photonic Quantum Systems (PhoQS), Paderborn University, Germany. Email: \{sevag.gharibian, jonas.kamminga\}@upb.de.} \and Jonas Kamminga\footnotemark[1]
}
\date{}
\begin{document}
\maketitle
\begin{abstract}
    Understanding the entanglement structure of local Hamiltonian ground spaces is a physically motivated problem, with applications ranging from tensor network design to quantum error-correcting codes. To this end, we study the complexity of estimating ground state entanglement, and more generally entropy estimation for low energy states and Gibbs states. We find, in particular, that the classes $\qqQAM$ [Kobayashi, le Gall, Nishimura, SICOMP 2019] (a quantum analogue of public-coin AM) and $\QMAt$ (QMA with unentangled proofs) play a crucial role for such problems, showing: (1) Detecting a ground state with high entropy on a subregister is $\qqQAM$-complete, (2) computing an additive error approximation to the Helmholtz free energy (equivalently, a multiplicative error approximation to the partition function) is in $\qqQAM$, (3) detecting  a low-entanglement ground state is $\QMAt$-hard, and (4) detecting low energy states which are close to  product states can range from QMA-complete to $\QMAt$-complete. Our results make progress on an open question of [Bravyi, Chowdhury, Gosset and Wocjan, Nature Physics 2022] on free energy, and yield the first $\QMAt$-complete Hamiltonian problem using \emph{local} Hamiltonians (\emph{cf.} the \emph{sparse} $\QMAt$-complete Hamiltonian problem of [Chailloux, Sattath, CCC 2012]).
\end{abstract}

\section{Introduction}

Estimating low energy properties of $k$-local Hamiltonian systems is a central challenge of quantum many-body physics. 
Here, a $k$-local Hamiltonian is an $n$-qubit Hermitian operator $H$ specified succinctly via $H=\sum_i H_i$, where each $H_i$ acts non-trivially on only $O(1)$ qubits. 
Such operators generalize classical Boolean $k$-local formulae, such as in $k$-SAT, and like their classical Boolean counterparts, $k$-local Hamiltonians are a cornerstone of quantum complexity theory~\cite{osborneHamiltonianComplexity2012,gharibianQuantumHamiltonianComplexity2015}. Briefly, low energy properties of $k$-local Hamiltonians are analogous to properties of the solution space to $k$-local Boolean formulae. 

The ``canonical'' low-energy property of $k$-local Hamiltonians is their \emph{ground state energy}, whose estimation is complete for Quantum Merlin Arthur\footnote{QMA is the canonical quantum generalization of NP~\cite{aharonovQuantumNPSurvey2002,gharibianGuestColumn72024}.} (QMA)~\cite{kitaevClassicalQuantumComputation2002,kempejulia3localHamiltonianQMAcomplete2003,kempeComplexityLocalHamiltonian2006}, and which generalizes the maximum number of satisfiable clauses of a Boolean formula.
Formally, the ground state energy of $H$ is its smallest eigenvalue $\lmin(H)$, which corresponds to the energy level the  system relaxes into when cooled to low temperature; this quantity plays an important role, e.g., in quantum chemistry in estimating reaction rates.
However, what is \emph{also} often important (e.g. for quantum error-correcting codes) are properties of the corresponding \emph{ground space} itself, i.e. eigenvectors/ground states $\ket{\psi}$ such that $H\ket{\psi}=\lmin(H)\ket{\psi}$. Ground states encode the state the system can take on at its ground energy.

To this end, numerous physically motivated low energy properties beyond ground state energy have been complexity theoretically characterized over the years, including ground state degeneracy~\cite{BFS11,y.shiNoteQuantumCounting}, simulation of local observables~\cite{ambainisPhysicalProblemsThat2014,gharibianComplexitySimulatingLocal2018,gharibianOracleComplexityClasses2020,gharibianPolynomiallyManyQueries2022,watsonComplexityTranslationallyInvariant2023,}, spectral gaps~\cite{ambainisPhysicalProblemsThat2014,cubittUndecidabilitySpectralGap2015,gharibianComplexitySimulatingLocal2018}, energy barriers~\cite{gharibianGroundStateConnectivity2018,gossetQCMAHardnessGround2017,gharibianQuantumSpaceGround2023,gharibianHardnessApproximationGround2024}, phase transitions~\cite{watsonComplexityApproximatingCritical2021,purcellChaitinPhaseTransition2024}, and ``universal'' Hamiltonians whose low energy space can reproduce the physics of other many-body systems~\cite{bravyiComplexityQuantumIsing2017,CMP18,kohlerTranslationallyInvariantUniversal2022}. Until recently, however, surprisingly absent from this list was the question: 
\begin{quote}
    \emph{What is the complexity of deciding whether $H$ has an entangled ground state?} 
\end{quote}
In this work, we will show that this question is more nuanced than one might expect, and draw connections to various complexity classes and related important concepts such as the free energy of Gibbs states.

\paragraph{Ground state entanglement: Prior work.} We begin by discussing three recent works which have studied the complexity of detecting ground state entanglement. The first two~\cite{gheorghiuEstimatingEntropyShallow2024,boulandPublicKeyPseudoentanglementHardness2024} achieve hardness for the Learning With Errors (LWE) problem~\cite{regevLatticesLearningErrors2009}, and the third~\cite{gharibianHardnessApproximationGround2024} for Quantum Classical Merlin Arthur (QCMA), i.e. QMA but with a classical proof. In more detail, Gheorghiu and Hoban~\cite{gheorghiuEstimatingEntropyShallow2024} first showed that given two local Hamiltonians and a pre-specified cut $A$ versus $B$ among qubits, it is LWE-hard to decide: Among all ground states $\ket{\psi_1}$ of $H_1$ and $\ket{\psi_2}$ of $H_2$ of minimal entanglement entropy, which of $\ket{\psi_1}$ or $\ket{\psi_2}$ has larger entanglement entropy\footnote{\emph{Entanglement entropy} is a standard entanglement measure, defined for a bipartite pure state $\psiab$ as $S(\rho_A)$, for reduced state $\rho_A:=\Tr_{B}(\rho_{AB})$ and $S(\rho):=-\Tr(\rho\log \rho$) the von Neumann entropy.} across $A$ versus $B$? Next, Bouland, Fefferman, Ghosh, Metger, Vazirani, Zhang and Zhou~\cite{boulandPublicKeyPseudoentanglementHardness2024,boulandHardnessLearningGround2024} showed that deciding whether the ground states of geometrically constrained $H$ have volume law or area law entanglement across a given set of cuts is LWE-hard. Note that the input model in~\cite{boulandPublicKeyPseudoentanglementHardness2024} is cryptographic in flavor, in that $H$ is drawn \emph{randomly} from one of two computationally indistinguishable families. Finally, in the standard complexity theoretic input model of a single $H$ being given as input, Gharibian and Hecht~\cite{gharibianHardnessApproximationGround2024}  showed: Given $H$, cut $A$ versus $B$, and 
thresholds $\eta_1<\eta_2$, it is QCMA-hard to distinguish whether $H$ has a ground state with entanglement entropy at most $\eta_1$ across the $A$ versus $B$ cut, or whether all ground states have entanglement entropy at least $\eta_2$ across this cut, \emph{even if} $\eta_2/\eta_1$ is large. Thus, entanglement entropy estimation is even QCMA-hard \emph{to approximate}.

\paragraph{Our starting point: Entropy estimation.} A challenge in studying ground state entanglement is that given a physical copy of a bipartite state $\psiab$, estimating entropy $S(\psi_A)$ is not known to be efficiently computable quantumly. This is why the careful reader may notice that, e.g., \cite{gharibianHardnessApproximationGround2024} does not show \emph{containment} in any complexity class, but only QCMA-hardness. 
Indeed, it is generally believed entropy estimation is \emph{not} in BQP. If $\psiab$ were specified via a poly-size quantum preparation circuit and given threshold $t$, deciding whether $S(\psi_A)$ is at least $t+1/2$ or at most $t-1/2$ is complete for Non-Interactive Quantum Statistical Zero Knowledge ($\NIQSZK$)~\cite{CCKV08}!
Note that while $\NIQSZK \subseteq \QSZK \subseteq \QIP(2)\subseteq\PSPACE$, for $\QIP(2)$ the class of decision problems with $2$-message\footnote{While $\QIP(3)=\QIP=\PSPACE$~\cite{kitaevParallelizationAmplificationExponential2000,jainQIPPSPACE2011}, $\QIP(2)$ remains a poorly understood class.} quantum interactive proofs, the relationship between $\QMA$ and $\NIQSZK$ is unclear. In fact, there is an oracle separation relative to which $\NIQSZK$ is not contained in $\PP$\cite{BCHTV19}.\footnote{In fact, \cite{BCHTV19} proves a stronger statement: there is an oracle $O$ such that $\NISZK^O \not\subseteq \UPP^O$.} 

In the case of ground state entanglement, we face the additional challenge that we do not even have a poly-size circuit preparing $\psiab$ (this would imply $\QMA=\QCMA$) --- this breaks even the QSZK containment proof of~\cite{BST08} for entropy estimation, raising various questions: \emph{What is the best upper bound on ground state entanglement estimation? Does the complexity change if we are interested in high or low entanglement ground states? And is there a connection to $\QMAt$, which is $\QMA$ with unentangled proofs?} 

\paragraph{Beyond the low energy setting: Free energy of Gibbs states.} Finally, our study will also take us beyond the low energy setting to \emph{high temperature} Gibbs states, where entropy estimation again plays an important role. Specifically, a \emph{Gibbs state} for a given Hamiltonian $H$ is the mixed state $\rhog=e^{-\bt H}/\calZ$, for normalization $\calZ:=\Tr(e^{-\bt H})$ denoting the \emph{quantum partition function}. Roughly, at temperature $T=1/\bt$, the quantum system thermalizes to $\rhog$, where the ground state is recovered by setting $T=0$. The relevant quantity for us is the \emph{Helmholtz free energy}, which plays an important role (e.g.) in mapping out phase diagrams, and which generalizes ground state energy to temperatures $T>0$. Specifically, the free energy of $H$ at inverse temperature $\bt$ is $\calF=\Tr(\rhog H) - (1/\bt)S(\rhog)$, where note the entropy $S(\rhog)$. 

Bravyi, Chowdhury, Gosset and Wocjan~\cite{bravyiQuantumHamiltonianComplexity2022} studied the natural question of additive error approximations to $\calF$ using an equivalent definition,
\begin{equation}
    \calF=-(1/\bt)\log \calZ.
\end{equation} 
Thus, an additive approximation of $\calF$ is equivalent to a relative approximation of $\calZ$, dubbed the Quantum Partition Function (QPF) problem in~\cite{bravyiQuantumHamiltonianComplexity2022}. Classically, although \emph{exactly} computing the partition function is $\#\textup{P}$-hard, relative approximation  is in $\BPP^{\NP}$ via Stockmeyer counting. In contrast, in the quantum case, it is only known that the problem is contained in $\p^{\#\p}$ \cite{bravyiQuantumHamiltonianComplexity2022,BFS11}, at least QMA-hard due to the special case $T=0$, and if $\bt$ is complex, then QPF is $\#\textup{P}$-hard within constant relative error (in fact, even for classical $H$~\cite{goldbergComplexityApproximatingComplexvalued2017,fujiiCommutingQuantumCircuits2017,galanisComplexityApproximatingComplexValued2020}). In this direction, \cite{bravyiQuantumHamiltonianComplexity2022} showed that QPF is equivalent under poly-time reductions to various problems, such as the number of eigenvalues of a $k$-local Hamiltonian in a given energy interval, but the complexity of these equivalent problems is also unfortunately unknown. This brings us to an open question\footnote{For clarity, there is a long history of study in relative approximations to $\calZ$, e.g. the landmark work of Jerrum and Sinclair~\cite{jerrumPolynomialTimeApproximationAlgorithms1993}. We cite \cite{bravyiQuantumHamiltonianComplexity2022} here, as it is the first paper to our knowledge to formally study the computational complexity of the general version of the problem, i.e. for all quantum Hamiltonians. However, special classes of quantum Hamiltonians have been considered before, e.g.~\cite{bravyiMonteCarloSimulation2015, bravyiPolynomialTimeClassicalSimulation2017, harrowClassicalAlgorithmsCorrelation2020,mannEfficientAlgorithmsApproximating2021}} of~\cite{bravyiQuantumHamiltonianComplexity2022}, which we also focus on here: \emph{What is the complexity of QPF, i.e. additive approximations to the free energy?}

\paragraph{Our results.} We study three computational problems involving low energy entanglement, as well as a fourth problem, free energy estimation. These are grouped below via the relevant complexity classes as follows: First, high entanglement low energy states ($\HELES$, \Cref{def:heles}) and free energy estimation ($\FEA$, \Cref{def:FEA}) are discussed under $\qqQAM$, as both involve entropy estimation of states for which no succinct preparation circuit is known, a task suited to $\qqQAM$. Next, low entanglement low energy states ($\LELES$, \Cref{def:leles}) and approximately product ground states ($\LEAPS$, \Cref{def:leaps}) are grouped under $\QMAt$. The distinction between $\LELES$ and $\LEAPS$ is the measure of entanglement employed --- entanglement entropy versus trace distance to a product state, respectively. Together, the aim is to highlight how seemingly minor changes to the definition of the problem being studied can drastically affect the relevant computational complexity. 

\paragraph{1. High entanglement low energy states and free energy: $\qqQAM$.} For our first pair of problems, an appropriate complexity class turns out to be an arguably under-the-radar class of Kobayashi, le Gall, and Nishimura~\cite{KGN19}, $\qqQAM$. Roughly, $\qqQAM$ refers to promise problems decidable by two-turn ``quantum public-coin'' interactive proof systems, in which the first message from verifier to prover consists only of polynomially many halves of $2$-qubit Bell states, i.e. ``quantum coin flips'' (\Cref{def:qqQAM}). An alternate, perhaps more natural viewpoint we observe here is that $\qqQAM$ can essentially be viewed as QMA, except instead of a pure state proof $\ket{\psi}$, the prover sends the \emph{normalized Choi matrix\footnote{The Choi matrix of a superoperator $\Phi:\calX\rightarrow\calY$ is defined as $J(\Phi)=\sum_{ij}\ketbra{i}{j}\otimes \Phi(\ketbra{i}{j})$. In other words, $\Phi$ is applied to half of an unnormalized high-dimensional Bell state. $\Phi$ is completely positive iff $J(\Phi)\succeq 0$, and $\Phi$ is trace-preserving iff $\Tr_\calY(J(\Phi))=I_\calX$.}} $J(\Phi)$ of some quantum channel as a proof\footnote{More accurately, after the prover's message, the verifier has proof $\rho_{ABC}$, so that $\rho_{AB}$ is the Choi state of some channel. 
}. Equivalently, the proof is a mixed state $\rho_{AB}$ with $\rho_A = \tilI$, for $\tilI$ the normalized identity.\\

\vspace{-1mm}
\noindent\emph{1.1 High entanglement low energy states.} With $\qqQAM$ in hand, the first problem we study is:
\begin{restatable}[High Entropy Low Energy State ($\HELES$)]{definition}{defHELES}\label{def:heles}
    Given is a $k$-local Hamiltonian $H_{AB} = \sum_{j = 1}^m H_j$ on registers $A$ and $B$ of $n_A + n_B = n$ qubits, with local terms satisfying $\|H_j\|_\infty \le 1$ and $H_j \succeq 0$ for all $j$, energy threshold parameters $\El, b$ such that $b - \El = 1/\poly(n)$ and entropy threshold parameters $s,t$ such that $s - t = 1/\poly(n)$. Decide:
    \begin{itemize}
        \item (YES case) If $\exists\; \psiab$ such that $\braketb{\psi}{H} \le \El$ and $S(\psi_A) \ge s$.
        \item (NO case) If $\forall\; \rho_{AB}$ either $\Tr(H \rho_{AB}) \ge b$ or $S(\rho_A) \le t$.
    \end{itemize}
\end{restatable}
\noindent Above and in this work, we use shorthand $\psi_A$ to denote $\Tr_B(\ketbra{\psi}{\psi}_{AB})$. In words, HELES asks: \emph{Does $H$ have a low energy state with high entropy on the $A$ register?} (We warn the reader already that whether the YES case corresponds to \emph{high} or \emph{low} entropy will change the complexity analysis completely, to be discussed shortly.) Our first result is as follows.
\begin{restatable}[Informal]{theorem}{thmHeles}\label{thm:heles}
    $\HELES$ is $\qqQAM$-complete for $k \ge 5$. Hardness remains even when restricting to physically motivated Hamiltonians such as the 2D Heisenberg model (\cref{cor:physicalHamiltonians}).
\end{restatable}
\noindent Thus, while entropy estimation is NIQSZK-complete when $\ket{\psi}_{AB}$ is specified via quantum circuit~\cite{kobayashiNoninteractiveQuantumPerfect2003,BST08}, the complexity jumps to $\qqQAM$ when $\ket{\psi}_{AB}$ is specified implicitly via a local Hamiltonian's low energy space. We say ``jumps'' because $\NIQSZK\subseteq\qqQAM$.

Note that as entanglement entropy is only an entanglement measure for pure states, $\HELES$ does not quite capture the question \emph{``Does $H$ have a low energy state with high entanglement entropy across the $A$ vs $B$ cut?''} $\HELES$ YES cases do indeed have such a low energy high entanglement entropy state, but in the NO case $\HELES$ has a stronger promise: not only does $H$ not have any low energy pure states with high entropy on the $A$ register (i.e., states with high entanglement entropy), but $H$ does not have any \emph{mixed} states with high entropy on the $A$ register either. Mixed low-energy states such as $\rho_A \otimes \rho_B$ with $S(\rho_A)$ high are thus not allowed either, even though these states are not entangled at all. Nevertheless, $\HELES$ serves a lower bound on the entanglement entropy question.
\\

\vspace{-1mm}
\noindent\emph{1.2 Free energy.} We next study the complexity of additive error approximations to free energy (equivalent to QPF of \cite{bravyiQuantumHamiltonianComplexity2022}).

\begin{restatable}[Free Energy Approximation ($\FEA$)]{definition}{defFEA}\label{def:FEA}
    Given is an $n$-qubit, $k$-local Hamiltonian $H = \sum_{j = 1}^m H_j$ with $\|H_j\| \le 1$ and $0 \preceq H_j$ for all $j$, an inverse temperature $\bt$, and two parameters $a, b$ satisfying $b - a = 1/\poly(n)$. Decide:
    \begin{itemize}
        \item (YES case) The free energy of $H$ at inverse temperature $\bt$ is less than $a$. That is $\calF(H) \le a$.
        \item (NO case) The free energy of $H$ at inverse temperature $\bt$ is larger than $b$. That is, $\calF(H) \ge b$. 
    \end{itemize}
\end{restatable}

\noindent We show:
\begin{restatable}{theorem}{thmFEA}\label{thm:freeenergy}
    For $k = O(1)$ and $\bt \ge \frac{1}{\poly(n)}$, $\FEA$ is contained in $\qqQAM$.
\end{restatable}
\noindent Thus, we obtain a bound on the power of additive approximations to the free energy, or equivalently, relative error approximations to a quantum partition function. It is unclear however, how $\qqQAM$ relates to the best previously known bound $\p^{\#\p}$.

Unlike the case of $\HELES$, we only manage to show containment in $\qqQAM$ for $\FEA$. This might seem surprising as the problems seem very similar: both combine low energy and high entropy. The reason for this difference is that our techniques for embedding (the Stinespring representation of) a $\qqQAM$-complete channel problem requires registers to be traced out, both for the Stinespring representation as well as to deal with the register added for the clock. In $\HELES$ this can be done by including those registers in the $B$ part, which is traced out before the entropy is considered. Contrarily, our definition of $\FEA$ does not allow such a thing, causing the difference. 

\paragraph{2. Low entanglement and approximately product ground states: $\QMAt$ and $\QMA$.} We tie the complexity of our next two problems to $\QMAt$~\cite{kobayashiQuantumMerlinArthurProof2003}, a class notorious for resisting anything but the trivial complexity bounds $\QMA\subseteq\QMAt\subseteq\NEXP$.\\

\vspace{-1mm}
\noindent\emph{2.1 Low entanglement low energy states.} Next, we consider the other side of the coin: low rather than high entropy/entanglement states. Our third problem is a seemingly innocuous variant of HELES, in which we seek a \emph{low} entanglement state (as opposed to a high entanglement one).

\begin{restatable}[Low Entropy Low Energy State ($\LELES$)]{definition}{defleles}\label{def:leles}
    Given is a $k$-local Hamiltonian $H$ on two registers $A$ and $B$ of $n_A + n_B = n$ qubits, energy thresholds $\El, b$ such that $b - \El = 1/\poly(n)$, and entropy thresholds $s,t$ such that $s - t \ge \frac{1}{\poly(n)}$. Decide:
    \begin{itemize}
        \item (YES case) If $\exists\; \psiab$ such that $\braketb{\psi}{H} \le \El$ and $S(\psi_A) \le t$.
        \item (NO case) If $\forall\; \rho_{AB}$ either $\Tr(H \rho) \ge b$ or $S(\rho_A) \ge s$.
    \end{itemize}
\end{restatable}
\noindent Note that $\LELES$ is \emph{not} the complement of $\HELES$, as we have only flipped the condition on entropy. Formally, LELES is equivalent to the ground state entanglement problem shown QCMA-hard to approximate in~\cite{gharibianHardnessApproximationGround2024}. Our initial aim was to show a non-trivial \emph{upper bound} on LELES, which proved difficult; the following theorem explains why.

\begin{restatable}[Informal]{theorem}{thmleles}\label{thm:leles}
    For $k \ge 5$, $\LELES$ is $\QMAt$-hard, even when $s - t = \Omega(n)$. Hardness remains even when restricting to physically motivated Hamiltonians such as the 2D Heisenberg model (\cref{cor:physicalHamiltonians}).
\end{restatable}
\noindent Thus, upper bounding the complexity of LELES is at least as difficult as bounding the power of $\QMAt$, a longstanding open problem. Whether this implies finding \emph{low} entanglement ground states is easier (LELES) or harder than finding \emph{high} entanglement ground states (HELES) depends on what one believes regarding the power of $\QMA$ versus $\QMAt$. It is entirely plausible that $\QMAt = \QMA \subsetneq \qqQAM$ (i.e. LELES is easier than HELES), or that $\QMAt = \NEXP \supsetneq \qqQAM$ (i.e. LELES is much harder than HELES).

Proving that $\LELES$ is also contained in $\QMAt$ faces a roadblock: approximating entropy is hard when one does not have a preparation circuit. In general, estimating the entropy of a unknown state $\rho$ up to constant accuracy requires an exponential number of copies of $\rho$ (linear in the dimension) \cite[Theorem~I.4]{WZ25}. Even though $\QMAt$ can receive polynomially many copies of the proof, this is a priori not enough to approximate the entropy sufficiently well.\\

\vspace{-1mm}
\noindent\emph{2.2 Low energy approximate product states.} As entropy is a challenging quantity, and is a barrier to showing $\QMAt$ containment of $\LELES$, our fourth and final problem studies a different measure of entanglement better suited for showing $\QMAt$ containment: trace distance to a product state. 

\begin{restatable}[Low Energy Approximate Product State ($\LEAPS$)]{definition}{defleaps}\label{def:leaps}
    Given is an $n$ qubit $k$-local Hamiltonian $H = \sum_{i = 1}^{m} H_i$ with $\|H_i\|_\infty \le 1$ for all $i$, energy thresholds $\El, b$, and distance thresholds $\dy, \dn < 1$. Decide:
    \begin{itemize}
        \item (YES case) There exists a low-energy state that is $\dy$-close to product in trace distance: $\exists \psiab, \ket{\phi_A}_A, \ket{\phi_B}_B$ such that $\braketb{\psi}{H} \le \El$ and $\|\ketbrab{\psi}_{AB} - \ketbrab{\phi_A} \otimes \ketbrab{\phi_B}\|_1 \le \dy$.
        \item (NO case) All low-energy states are $\dn$-far from product: $\forall \psiab, \ket{\phi_A}_A, \ket{\phi_B}_B$, either $\braketb{\psi}{H} \ge b$ or $\|\ketbrab{\psi}_{AB} - \ketbrab{\phi_A} \otimes \ketbrab{\phi_B}\|_1 \ge \dn$.
    \end{itemize}
\end{restatable}

\noindent This problem is related\footnote{A further related work is \cite{kallaugherComplexityClassificationProduct2025} (see also \cite{piddockQuantumMaxCutNP2025}), which  also study the complexity of product state energy optimizations. These differ from the present work and from~\cite{CS12} in that in \cite{kallaugherComplexityClassificationProduct2025}, ``product state'' refers to tensor products of \emph{single} qubits, whereas here we have a single bipartition of all $n$ qubits. Thus, \cite{kallaugherComplexityClassificationProduct2025}'s complexity classification involves NP, which is not the case here and in~\cite{CS12}.} to the work of Chailloux and Sattath~\cite{CS12}, which showed an intriguing separation: On the one hand, estimating $\min_{\ket{\psi}\ket{\phi}}\Tr(H\ketbra{\psi}{\psi}\otimes \ketbra{\phi}{\phi})$ for \emph{local} Hamiltonians $H$ (denoted $\cfont{Separable Local Hamiltonian}$~\cite{CS12}) is \QMA-complete, despite the tensor product requirement $\ket{\psi}\ket{\phi}$, which might \emph{a priori} be associated with $\QMAt$. On the other hand, allowing non-local but \emph{sparse} $H$ (denoted $\cfont{Separable Sparse Hamiltonian}$~\cite{CS12}) yields $\QMAt$-completeness. While sparse Hamiltonians are well motivated (e.g. they can be efficiently simulated quantumly~\cite{aharonovAdiabaticQuantumComputation2004}), in terms of many-body systems one is often interested in \emph{local} Hamiltonians. We show:
\begin{restatable}{theorem}{thmleapsqma}\label{thm:leapsQMA}
    Let $D$ be an efficiently computable upper bound on $\|H\|_\infty$ such as the number of local terms $m$. Then $\LEAPS$ is $\QMA$-complete if $b - \El - \dy D \ge \frac{1}{\poly(n)}$ and $\dn > \dy$. Note that no (inverse polynomial) gap between $\dy$ and $\dn$ is required.
\end{restatable}
\noindent No gap between $\dy$ and $\dn$ is required. This might seem strange, but mirrors $\cfont{Separable Local Hamiltonian}$. There YES cases are those where a low-energy product state exists (i.e. distance $0$ from product), whereas NO cases have no low-energy product states. That is, in the NO case, all low energy states have non-zero distance from product. Alternatively, as energy varies continuously with distance, one can think of the distance gap as being subsumed by the energy gap.

However, for a different parameter regime we show:
\begin{restatable}{theorem}{thmleapsqmat}\label{thm:leapsQMA2}
    For all constant $k\ge 5$, $\LEAPS$ is $\QMAt$-complete if $b - \El = \frac{1}{\poly(n)}$, $\dn - \dy = \frac{1}{\poly(n)}$, $\dy = \frac{1}{\poly(n)}$ and $b = O(\dy^6)$. Hardness remains even when restricting to physically motivated Hamiltonians such as the 2D Heisenberg model (\cref{cor:physicalHamiltonians}).
\end{restatable}
\noindent Note that $\|H\|_\infty \ge 1$. Therefore, our $\QMA$-containment result requires $b > \dy$, whereas we only show $\QMAt$-hardness when $b = O(\dy^6)$. In particular, as $\dy < 1$, there is no overlap between the regimes where we show $\QMA$ and $\QMAt$-completeness.

In sum, we obtain a natural \emph{local} Hamiltonian problem whose complexity ``leaps'' from $\QMA$ to $\QMAt$, depending on the choice of parameters. In the words of an anonymous referee (whom we thank), this ``restores [one's] faith'' in the standard intuition that locality of a Hamiltonian (meaning $k$-local versus non-local but sparse) should not affect the quantum complexity of low energy problems. Finally, LEAPS is additionally appealing for our context as it studies not just exact product states (as in~\cite{CS12}), but the more robust notion of \emph{close to product} states, which corresponds to low entropy. 

A summary of our results is given in \cref{Table:results}.

\begin{center}
    \begin{table}
    \small
    \SetTblrInner{rowsep=6pt}
    \begin{tblr}{colspec={X[0.9,c]X[3,l]X[3,l]X[1.7,l]}}
        Problem & YES condition & NO condition & Result \\
        \hline
        $\HELES$ & $\exists$ a low energy state with \emph{high} entropy on subsystem $A$:\newline $\exists \ket{\psi}_{AB}$ s.t. $\braketb{\psi}{H} \le \El$ and $S(\psi_A) \ge s$ & No state has both low energy and \emph{high} entropy on subsystem $A$:\newline $\forall \rho_{AB}$, $\Tr(H\rho_{AB}) \ge b$ or $S(\rho_A) \le t$ & $\qqQAM$-complete (\ref{Lem:HELEScontainment},\ref{Lem:HELEShard})\\ 
        \hline
        $\LELES$  & $\exists$ a low energy state with \emph{low} entropy on subsystem $A$:\newline $\exists \ket{\psi}_{AB}$ s.t. $\braketb{\psi}{H} \le \El$ and $S(\psi_A) \le t$ & No state has both low energy and \emph{low} entropy on subsystem $A$:\newline $\forall \rho_{AB}$, $\Tr(H\rho_{AB}) \ge b$ or $S(\rho_A) \ge s$  & $\QMAt$-hard (\ref{thm:leles}) \\
        \hline
        $\FEA$ & The free energy of $H$ is low: \newline $\calF(H) \le a$ & The free energy of $H$ is high: \newline $\calF(H) \ge b$ & Contained in $\qqQAM$ (\ref{thm:freeenergy}) \\
        \hline
        $\LEAPS$ & $\exists$ a low energy state that is $\dy$ close to product:\newline $\exists \ket{\psi}_{AB}, \ket{\phi_L},\ket{\phi_R}$ s.t. $\braketb{\psi}{H} \le \El$ and $\|\psi -  \phi_L \otimes \phi_R\|_1 \le \dy$ & All low energy states are $\dn$ far from product:\newline $\forall \ket{\psi}_{AB}, \ket{\phi_L},\ket{\phi_R}$ either $\braketb{\psi}{H} \ge b$ and $\|\psi -  \phi_L \otimes \phi_R\|_1 \ge \dn$ & $\QMA$-complete when $\dy \ll b - \El$ (\ref{lem:LEAPSQMAcontained},\ref{lem:LEAPSQMAhard}), $\QMAt$-complete when $\dy \gg b$ (\ref{lem:LEAPSQMA2contained},\ref{lem:LEAPSQMA2hard}) 
    \end{tblr}
    \normalsize
    \caption{A summary of the problems we consider and the results we obtain.}
    \label{Table:results}
    \end{table}
\end{center}

\vspace{-1mm}
\paragraph{Techniques.} We make use of two main technical tools, which we encapsulate into two standalone lemmas (\Cref{Lem:EntropyVerification} and \Cref{Lem:ChanneltoHamiltonian}, respectively) for ease of further adoption. The first is a $\qqQAM$ protocol for entropy verification, which we use to prove that $\qqQAM$ contains $\HELES$ and $\FEA$. This protocol is not new; it appears implicitly in \cite[proof of Lemma 33]{KGN19} and is based on ideas from \cite[Section~5]{BST08}. The second tool is a ``channel-to-Hamiltonian'' construction: Kitaev's circuit-to-Hamiltonian construction applied to the Stinespring representation of a channel with minor modifications. We use this to prove that $\HELES$ is $\qqQAM$-hard and to prove that $\LELES$ and $\LEAPS$ are $\QMAt$-hard. We now sketch each of these tools.\\

\vspace{-1mm}
\noindent\emph{Entropy verification.} The workhorse of this protocol is a quantum extractor, i.e. a channel $T$ roughly satisfying $T(\sigma)\approx \tilI$ if and only if the min-entropy\footnote{The min-entropy is defined as $H_\infty(\sigma) = - \log \|\sigma\|_\infty$ and lower bounds the usual von Neumann entropy.} $H_\infty(\sigma)$ is sufficiently large (by $\tilI$ we denote the maximally mixed state). We will use the Stinespring dilation unitary $U_T$ of $T$, which satisfies $T(\sigma) = \Tr_E\left(U_T (\sigma \otimes \ketbrab{0^{2n_A}}_E) U_T^\dagger\right)$. To help understand the protocol, one can adopt the view of $\qqQAM$ as QMA with a normalized Choi matrix $\rho_{AB}$ as proof. Recall this gave us the key property that $\rho_A=\tilI_A$. Note that by unitary freedom of purifications, the prover can prepare any desired purification of the $\tilI_A$ system. 

The entropy verification procedure will be a sequence of actions by the verifier transforming the proof $\rho_{AB}$ into a different state $\sigma_A$ such that
\begin{enumerate}
    \item \emph{Any high entropy state can be achieved:} For any high entropy state $\chi_A$, a cooperating prover can send an appropriate proof $\rho_{AB}$ such that the entropy verification accepts with probability $\ccomp$, in which case the output $\sigma_A$ is close to $\chi_A$.
    \item \emph{Cheating provers will be caught:} If the prover tries to make the entropy verification procedure output a low entropy state, then the procedure accepts with probability at most $\csound$. In other words, if the output state $\sigma_A$ in case of acceptance has low entropy, then the acceptance probability must have been smaller than $\csound$.
\end{enumerate}
Furthermore, the completeness and soundness parameters satisfy $\ccomp - \csound = \frac{1}{\poly}$.

The entropy verification procedure thus allows the verifier to be confident that the output state $\sigma$ is high-entropy (in case of acceptance). Furthermore, a cooperating prover (who wants to send high entropy states) is not restricted by the procedure.

The procedure works as follows (\Cref{Fig:EntropyVerification}). First, the verifier ``undoes'' $T$. That is, they apply $U_T^\dagger$ to the entire proof state $\rho_{AB}$. Next, they measure $B$ in the computational basis and reject if the result is not all $0$'s. If this test passes with high probability, then the postmeasurement state $\sigma_A$ on register $A$ has high min-entropy. The reason for this is as follows. Suppose the probability of getting all $0$'s was high. Then, the state after applying $U_T^\dagger$ is approximately $\sigma_A \otimes \ketbrab{0^{2n_A}}_B$. Note that $U_T (\sigma_A \otimes \ketbrab{0^{2n_A}}_B) U_T^\dagger \approx \rho_{AB}$. In particular, if one would ``redo'' $T$, one gets $T(\sigma) = \Tr_B(U_T (\sigma_A \otimes \ketbrab{0^{2n_A}}_B) U_T^\dagger) \approx \tilI_A$. Since $T$ is an extractor, this implies the min-entropy of $\sigma_A$ is large.

There is one more hurdle to overcome: we need to move from \emph{min-entropy} to von Neumann entropy. Note that the min-entropy lower bounds the von Neumann entropy, so if the protocol accepts we can be sure that $\sigma_A$ has high von Neumann entropy as well. What goes wrong is that if $S(\sigma_A)$ is high, that does not imply that $T(\sigma) \approx \tilI$. Therefore, there are states that do have high von Neumann entropy, but cannot be the output of the protocol (when the measurement accepts with high probability). To fix this, we use a flattening lemma, which states that $H_\infty(\rho^{\otimes q}) \approx qS(\rho)$. It follows that is $S(\sigma)$ is high, then $\sigma^{\otimes q}$ can be the output of the protocol. Vice-versa, if $\sigma_{A^{\otimes q}}$ is the output of the protocol, the verifier knows that $S(\sigma_{A^{\otimes q}})$ is high. While the different copies of the $A$ register could be arbitrarily entangled, subadditivity of von Neumann entropy ensures that most of the reduced state $\sigma_{A_i}$ must also have high entropy. Curiously, this is where the potentially large complexity difference between $\HELES$ and $\LELES$ could come from: by subadditivity, the prover \emph{cannot} entangle \emph{low entropy} reduced states $\sigma_{A_i}$ to create a \emph{high entropy} global state $\sigma_{A^{\otimes q}}$, but they \emph{can} entangle \emph{high entropy} reduced states to create a \emph{low entropy} global state. This allows the prover to cheat in the NO case if one tries to run the entropy verification protocol to decide $\LELES$ in $\qqQAM$.\\

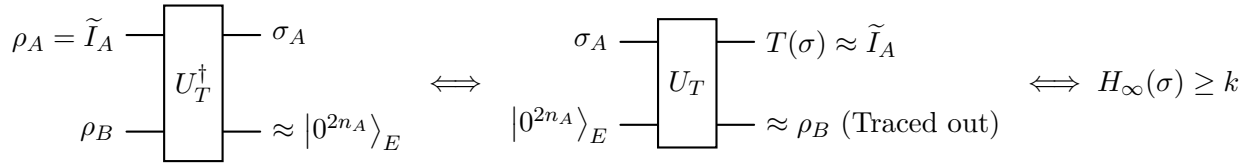
\begin{figure}
    \begin{center}
        \begin{quantikz}
        \lstick{$\rho_A = \tilI_A$} & \gate[2]{U_T^\dagger} & \rstick{$\sigma_A$} \\
        \lstick{$\rho_B$} &  &  \rstick{$\approx \ket{0^{2n_A}}_E$}
    \end{quantikz}
    $\iff$
    \begin{quantikz}
        \lstick{$\sigma_A$} & \gate[2]{U_T} & \rstick{$T(\sigma) \approx \tilI_A$} \\
        \lstick{$\ket{0^{2n_A}}_E$} &  &  \rstick{$\approx \rho_B$ (Traced out)}
    \end{quantikz}
    $\iff H_\infty(\sigma) \ge k$
    \end{center}
    \caption{Schematic representation of the entropy verification protocol. If the proof is such that the $E$ register after applying $U_T^\dagger$ is approximately $|0^{2n_A}\rangle$, then the contents of the $A$ register $\sigma_A$ are such that $T(\sigma_A) \approx \tilI$. As $T$ is an extractor, this implies that the min-entropy of $\sigma_A$ is large.}
    \label{Fig:EntropyVerification}
\end{figure}

\vspace{-1mm}
\noindent\emph{Channel-to-Hamiltonian construction.} We next discuss our \emph{channel-to-Hamiltonian} construction, which  allows us to reduce hardness results for \emph{channel} problems to hardness results for Hamiltonians. Specifically, we use it to prove that $\HELES$ is $\qqQAM$-hard and that $\LELES$ and $\LEAPS$ are $\QMAt$-hard. This construction bypasses the ``sparse Hamiltonian roadblock'' of \cite{CS12}, allowing us to obtain $\QMAt$-hardness for $\LELES$ and $\LEAPS$ even for \emph{local} $H$. This is because it allows us to reduce from the channel-based $\QMAt$-complete Pure Product Isometry Output (PPIO) problem of Gutoski, Hayden, Milner and Wilde~\cite{GHMW13}, rather than utilizing the Product State Test of Harrow and Montanaro~\cite{HM13} as done in \cite{CS12}, which is a non-local operation.

The idea behind the channel-to-Hamiltonian construction is as follows. Given a channel $\Phi$, we wish to define a Hamiltonian whose ground space consists of all states the channel could output. To achieve this, we use Kitaev's circuit-to-Hamiltonian construction~\cite{kitaevClassicalQuantumComputation2002} applied to the Stinespring dilation unitary $U_\Phi$, which recall satisfies $\Phi(\rho) = \Tr_E(U_\Phi (\rho \otimes \ketbrab{0^{2n_A}}_E) U_\Phi^\dagger)$. We make two modifications: we do not use an output penalty and we post idle the computation.

As we do not use an output penalty, the zero-energy space of the resulting Hamiltonian $H_\Phi$ will be the span of \emph{history states} of the form 
\begin{align}
    \ket{\psi_{\textup{hist}}} &= \frac{1}{\sqrt{T + 1}}\sum_{t = 0}^{T} V_t \dots V_1 \ket{\psi}_A \ket{0^{n_B}}_B \ket{t}_C,
\end{align}
where $U_\Phi = V_T \dots V_1$ and $C$ is the clock register. Note that if we trace out the $B$ and $C$ registers from the final time step we get exactly $\Phi(\ketbrab{\psi})$, but it is less clear what happens to the earlier time steps. To make the history state after tracing out look more like $\Phi(\ketbrab{\psi})$, we use post idling: we add $L$ more time steps, for which $V_t = I$. Thus:
\begin{align}
    \ket{\psi_{\textup{hist}}} = \frac{1}{\sqrt{T + L +1}}\sum_{t = 0}^{T-1} V_t \dots V_1 \ket{\psi}_A \ket{0^n}_B \ket{t}_C+ U_\Phi\ket{\psi}_A\ket{0^n}_B \left(\frac{1}{\sqrt{T + L + 1}} \sum_{t = T}^{T+L}\ket{t}_C\right).
\end{align}
Tracing out $B$ and $C$, the resulting state approximates $\Phi(\ketbrab{\psi})$.\\

\vspace{-1mm}
\noindent\emph{``Leaping'' from $\QMA$ to $\QMAt$.} The crucial insight that allows us to prove that $\LEAPS$ ``leaps'' from $\QMA$-completeness to $\QMAt$-completeness is as follows. Consider a state $\psiab$ that is $\dy$ close to $\ket{\phi_L}\ket{\phi_R}$ and has low energy. Suppose the $\dy$ parameter is substantially smaller than the gap $b - \El$. Then, the energy of $\ket{\phi_L}\ket{\phi_R}$ will not be too high, specifically it will be smaller than (say) $\El + \frac{b - \El}{4}$. To prove that $\psiab$ is low energy and close to product, the prover can thus send the state $\ket{\phi_L}\ket{\phi_R}$, which the techniques from \cite{CS12} allow them to do. The verifier can then check that $\ket{\phi_L}\ket{\phi_R}$ has energy $\le \El + \frac{b - \El}{4}$ which will convince them that they are dealing with a $\LEAPS$ YES case. Note that in the NO case, in particular all product states have energy at least $b$, so the prover cannot cheat. In other words, if $\dy$ is sufficiently smaller than $b$, the energy and distance checks can be combined into one.

Suppose now that $b$ is much smaller than $\dy$. In this case it is entirely possible that $\psiab$ has low energy and is $\dy$ close to product, but that any product state it is close to has energy $\ge b$. The energy and distance checks can thus not be combined like before. Instead, they need to be checked separately. This can indeed be done in $\QMAt$: the provers can send $\psiab$ and the product state $\ket{\phi_L}\ket{\phi_R}$ close to it. As the verifier is promised that these are all in product, they can check the distance with a SWAP test. Intuitively, the reason for $\QMAt$ hardness is that the verifier needs to verify this product structure.

\paragraph{Open questions.} This work has singled out two classes, $\qqQAM$ and $\QMAt$, as being important to the study of entanglement of low energy states and free energy of Gibbs states. Some problems we have been able to characterize completely, such as $\qqQAM$-completeness of high entropy ground state detection (\Cref{thm:heles}) and $\QMAt$-completeness of low energy approximate product states verification (\Cref{thm:leapsQMA2}). On the other hand, the three main open questions we leave are: (1) What is a tight upper bound on the complexity of low energy ground state detection (LELES)? On the one hand, showing a bound such as PSPACE would entail a major breakthrough in quantum complexity theory, due to our $\QMAt$-hardness result (\Cref{thm:leles}). On the other hand, a ``lower'' bar might be to show containment in $\QMAt$ itself, though it is unclear how one would verify entropy in $\QMAt$. (2) Can one show a lower bound of $\qqQAM$ on the complexity of additive error free energy/relative error partition function approximation, matching our containment result in $\qqQAM$ (\Cref{thm:freeenergy})? A potential obstacle here is that we do not have a particularly clean method of reducing hard problems to Gibbs states (as opposed to low energy states, for which we can use Kitaev's circuit-to-Hamiltonian construction). 
(3) What is the complexity of $\LEAPS$ between where we show $\QMA$-completeness and where we show $\QMAt$-completeness? That is, is there a ``sharp'' transition between $\QMA$ and $\QMAt$ or are there values where the problem is complete for some class between $\QMA$ and $\QMAt$?

\paragraph{Organization.} We begin with preliminaries in \Cref{scn:preliminaries}. \Cref{scn:lemmas} fleshes out our main two technical tools. \Cref{scn:heles} studies HELES, \Cref{scn:FEA} for FEA, \Cref{scn:leles} for LELES, and \Cref{scn:leaps} for LEAPS. Finally, we consider physical Hamiltonians in \cref{sec:physicalHamiltonians}.

\section{Preliminaries}\label{scn:preliminaries}

\paragraph{Notation.} We use $\|\ket{\psi}\|_2$ to denote the Euclidean norm of $\ket{\psi}$, $\|\rho\|_1 = \Tr(\sqrt{\rho \rho^\dagger})$ the trace norm of $\rho$ and $\|A\|_\infty=\max_{\ket{\psi}} \|A\ket{\psi}\|_2/\|\ket{\psi}\|_2$ the spectral/operator norm of $A$. 

Throughout this paper, $\tilI \coloneq \frac{1}{d}I$ refers to the maximally mixed state on a Hilbert space of dimension $d$. The corresponding Hilbert space will either be clear from context, or denoted in a subscript, e.g., $\tilI_A$ for the maximally mixed state on $\calH_A$. 

When $\ket{\psi}$ is a pure state, we will use the shorthand notation $\psi \coloneq \ketbrab{\psi}$ for the corresponding density matrix. Note that the symbols $\psi, \phi$ are used for pure states only.

\subsection{Complexity classes and problems}

\paragraph{The class $\qqQAM$.}
One of the main complexity classes used in this work is the relatively obscure $\qqQAM$, first defined in \cite{KGN19}. This class is a quantum generalization of the classical $\AM$ but is potentially different from $\QAM$, the more common quantum generalization of $\AM$. We now sketch the difference between $\qqQAM, \QAM$ and $\AM$. 

$\AM$ can be defined as the class of languages accepted by protocols of the following form: First, a $\BPP$ verifier Arthur and a prover Merlin share $\poly$ many bits of randomness. Merlin then prepares a proof based on this shared randomness and sends it to Arthur, who verifies the proof in combination with the random bits. In YES cases Arthur should accept with probability $2/3$ over the shared randomness, whereas he may accept NO cases only with probability at most $1/3$.

The class $\QAM$ generalizes $\AM$ to the quantum setting by making Arthur a $\BQP$ verifier and the proof a quantum state. The shared randomness, however, is still classical. Here $\QAM$ differs from $\qqQAM$. In $\qqQAM$, the classical correlations between Arthur and Merlin (the shared randomness) are replaced by quantum correlations. That is, Arthur and Merlin share $\poly$ many EPR pairs instead of random bits. They could measure these in the computational basis to obtain shared randomness, showing that $\QAM \subseteq \qqQAM$, but are not required to do so. Formally, we define $\qqQAM$ as follows.

\begin{definition}[$\qqQAM$-protocol]
    A $\qqQAM$-protocol consists of a verifier and an all-powerful prover interacting. On an input $x$ of length $n$ they take the following actions.
    \begin{enumerate}
        \item The verifier prepares $\poly(n)$ EPR pairs across registers $A, B_1$. They send the $B_1$ register to the prover and keep the $A$ register themselves.
        \item The prover performs any quantum operation on $B = B_1B_2$ and potentially on a private register $P$. They then send the $B$ register to the verifier who now has access to the proof state $\rho_{AB}$.
        \item The verifier applies a $\poly(n)$-size quantum circuit $U_x$ to $\rho_{AB}$ and potentially to a $\poly(n)$ size workspace register $W$. They then measure the first qubit of $A$ in the computational basis, accept if the result is 1 and reject otherwise.
    \end{enumerate}
\end{definition}

\begin{definition}[quantum-quantum Quantum Arthur Merlin ($\qqQAM$) see also \cite{KGN19}]\label{def:qqQAM}
    A promise problem $(L_{yes}, L_{no})$ is in $\qqQAM(c,s)$ if there is some $\qqQAM$-protocol such that
    \begin{itemize}
        \item $\forall x \in L_{yes}$, the verifier accepts with probability at least $c$.
        \item $\forall x \in L_{no}$, the verifier accepts with probability at most $s$.
    \end{itemize}
    We define $\qqQAM = \qqQAM(\frac{2}{3}, \frac{1}{3})$.
\end{definition}

In \cite[Lemma~20]{KGN19} it is proven that $\qqQAM$ completeness and soundness parameters can be boosted using parallel repetition. That is, for $c,s$ satisfying $c - s = \frac{1}{\poly}$ and for any polynomial $p$, $\qqQAM(c,s) = \qqQAM = \qqQAM(1- 2^{-p}, 2^{-p})$.

Note that by \cref{Lem:Purificationsrelatedbyunitary}, the prover can make $\rho_{AB}$ be \emph{any} state that satisfies $\rho_A = \tilI_A$. We can thus equivalently define $\qqQAM$ as $\QMA$ where the proof $\rho_{AB}$ is potentially mixed but promised to satisfy $\rho_{A} = \tilI_A$.

In \cite{KGN19}, it was shown that $\qqQAM$ has two natural complete problems.

\begin{definition}[Close Image to Maximally Mixed problem (CIMM)]
    The input to a $\CIMM$ instance consists of a quantum channel $\Phi$ (given as a circuit implementing its Stinespring representation) and constants $a,b \in (0,1)$ such that $(1-a)^2 > 1-b^2$. The task is to distinguish the following two cases:
    \begin{itemize}
        \item[YES:] $\exists \rho$ such that $\|\Phi(\rho) - \tilI\|_1 \le a$.
        \item[NO:] $\forall \rho$, $\|\Phi(\rho) - \tilI\|_1 \ge b$.
    \end{itemize} 
\end{definition}
\begin{theorem}[\cite{KGN19}]
    $\CIMM$ is complete for $\qqQAM$.
\end{theorem}

\begin{definition}[Maximum Output Quantum Entropy Approximation ($\MaxOutQEA$)]
    The input to a $\MaxOutQEA$ instance consists of (the Stinespring representation of) a quantum channel $\Phi$ and a parameter $t\in \RR$. The task is to distinguish the following two cases:
    \begin{itemize}
        \item[YES:] $\exists \rho$ such that $S(\Phi(\rho)) \ge t + 1$.
        \item[NO:]  $\forall \rho$, $S(\Phi(\rho)) \le t - 1$.
    \end{itemize}
\end{definition}
\begin{theorem}[{\cite[Theorem~3]{KGN19}}]
    $\MaxOutQEA$ is $\qqQAM$ complete.  
\end{theorem}

\paragraph{Pure Product Isometry Output.}
The following $\QMAt$-complete problem due to \cite{GHMW13} will be very useful for our purposes.
\begin{definition}[Pure Product Isometry Output ($\PPIO$)]\label{def:PPIO}
    The input to a $\PPIO$ instance is a quantum circuit implementing an isometry $U$ with a 2-partite\footnote{The general definition allows $\ell$-partite output systems, but for our purposes the 2-partite case suffices.} output system $AB$ and two parameters $\dy, \dn$. The task is to decide between:
    \begin{itemize}
        \item[YES:] There is some isometry output that is close to product: $$\exists \ket{\psi}, \ket{\phi_A}, \ket{\phi_B}\text{ such that }\|U\left(\ketbrab{\psi}_A\otimes \ketbrab{0^{n_B}}_B\right)U^\dagger - \ketbrab{\phi_A} \otimes \ketbrab{\phi_B}\|_1 \le \dy.$$
        \item[NO:] All isometry outputs are far from product: $$\forall \ket{\psi}, \ket{\phi_A}, \ket{\phi_B}\text{, }\|U\left(\ketbrab{\psi}_A\otimes \ketbrab{0^{n_B}}_B\right)U^\dagger - \ketbrab{\phi_A} \otimes \ketbrab{\phi_B}\|_1 \ge \dn.$$
    \end{itemize}
\end{definition}
\begin{theorem}[{\cite[Theorem~6.2]{GHMW13}}]
    $\PPIO$ is $\QMA(2)$-complete, also for $\dy = 2^{-\Omega(n)}$ and $\dn = 2 - 2^{-\Omega(n)}$.
\end{theorem}

\subsection{Useful definitions and lemmas}
\paragraph{Quantum entropies.}
In this paper we make use of two different quantum entropies: the von Neumann entropy and the quantum min-entropy. The von Neumann entropy is the quantum equivalent of the classical Shannon entropy and is defined as:
\begin{equation}
    S(\rho) = -\Tr\left(\rho \log \rho \right).
\end{equation}
Indeed, if $\rho$ is a classical probability distribution, the Shannon entropy and the von Neumann entropy coincide. 

The quantum min-entropy is defined as:
\begin{equation}
    H_\infty(\rho) = - \log \|\rho\|_\infty.
\end{equation}

Von Neumann entropy and min-entropy are related as follows:
\begin{equation}
    \forall \rho, \quad S(\rho) \ge H_\infty(\rho)
\end{equation}

\paragraph{Stinespring representation of channels.}
By the Stinespring Dilation Theorem (see e.g. \cite[Section~5.2.2]{Wil13}), any quantum channel $\Phi\colon L(\calH_A) \to L(\calH_B)$ can be written as
\begin{align}
    \Phi(\rho) = \Tr_E(V\rho V^\dagger),
\end{align}
where $V \colon \calH_A \to \calH_B \otimes \calH_E$ is an isometry and $E$ is an auxiliary register with $\dim(\calH_E) \le \dim(\calH_A) \dim(\calH_B)$. This way of writing $\Phi$ is often called the \emph{Stinespring representation} of $\Phi$. On finite dimensional Hilbert spaces (the setting we consider), isometries can be extended to unitaries. It follows that, for every channel $\Phi$, there exists a unitary $U_\Phi$ acting on the $A, B$ and $E$ registers such that 
\begin{align}
    \Phi(\rho) = \Tr_{AE}\left( U_\Phi(\rho \otimes \ketbrab{0^{n_B}}_B \otimes \ketbrab{0^{n_A + n_B}}_E) U_\Phi^\dagger \right).
\end{align} 
In this work, whenever we refer to the Stinespring representation of a channel, we refer to the unitary $U_\Phi$.

\paragraph{Quantum extractors.}
One tool we will use are quantum extractors. Informally, these are quantum channels that ``extract'' the entropy from their input state, outputting a maximally mixed state if the input has sufficient entropy and something far from maximally mixed if it does not. More formally:
\begin{definition}[Quantum extractors (see e.g. {\cite[section~5.1]{BST08}})]
    \label{Def:Extractor}
    A quantum channel $T\colon L(\calH) \to L(\calH)$ is a $(k,d,\epsilon)$-quantum extractor if it satisfies the following properties:
    \begin{enumerate}
        \item $T$ is $2^d$-regular. That is, there are unitary operators $\{U_i\}_{1 \le i \le 2^d}$ such that $T$ can be written as 
        \begin{equation}
            T(\rho) = \frac{1}{2^d} \sum_{i = 1}^{2^d} U_i \rho U_i^\dagger
        \end{equation}
        \item For every $\rho$ with $H_\infty(\rho) \ge k$, one has $\|T(\rho) - \tilI\|_1 \le \epsilon$.
    \end{enumerate}
    A quantum extractor is called explicit if it can be implemented by a quantum circuit of size $\poly(\dim \calH)$. 
\end{definition} 

\begin{lemma}[{\cite[Corollary~5.3]{BST08}}]
    \label{Lem:ExpandersExist}
    For all $n, k, \epsilon$, there exist constructions of explicit $(k,d,\epsilon)$-quantum extractors $T \colon L((\CC^2)^{\otimes n}) \to L((\CC^2)^{\otimes n})$ with $d = 2(n - k + 2 \log \frac{1}{\epsilon}) + O(1)$, or $d = n - k + 2 \log \frac{1}{\epsilon} + 2 \log n + O(1)$, depending on the construction used.
\end{lemma}

\paragraph{Swap test.}
The SWAP test (see e.g. \cite{BCWW01}) is a quantum subroutine that, when given as input a quantum state $\ket{\phi}\otimes \ket{\psi}$ accepts with probability $p_{\textup{acc}} = \frac{1}{2} + \frac{1}{2}|\braket{\psi}{\phi}|^2$. It can therefore be used to estimate the overlap of two pure states $\ket{\phi}$ and $\ket{\psi}$. Note that for pure states we have
\begin{align}
    \|\phi - \psi\|_1 &= 2\sqrt{1 - |\braket{\psi}{\phi}|^2} \\
    &= 2\sqrt{2 (1 -p_{\textup{acc}})}.
\end{align}
Therefore, the SWAP test can also be used to estimate the trace distance of two pure states:

\paragraph{Useful lemmas.}

\begin{lemma}[{Gentle Measurement Lemma (see e.g. \cite[Lemma~9.4.1]{Wil13})}]
    \label{Lem:GentleMeasurement}
    Let $\rho$ be a quantum state and consider a measurement operator $0 \preceq M \preceq I$ with $\Tr(M \rho) \ge 1 - \epsilon$. Then the post-measurement state $\rho_{post} = \frac{\sqrt{M}\rho \sqrt{M}}{\Tr(M \rho)}$ satisfies:
    \begin{equation}
        \|\rho - \rho_{post}\|_1 \le 2\sqrt{\epsilon}
    \end{equation}
\end{lemma}

\begin{theorem}[Uhlmann's Theorem]
    Let $\sigma$ and $\rho$ be two mixed states. Then
    \begin{equation}
        F(\sigma, \rho) = \max_{\ket{\phi_\sigma}, \ket{\phi_\rho}} |\braket{\phi_\sigma}{\phi_\rho}|^2
    \end{equation}
\end{theorem}

\begin{lemma}[{See \cite[Exercise~2.81]{NC10}}]
    \label{Lem:freedomofpurifications}
    Let $\ket{\phi}_{AB}$ and $\psiab$ be two purifications of a mixed state $\rho_A$. Then there is a unitary $U_B$ acting only on the purifying register such that $I_A \otimes U_B \ket{\phi} = \ket{\psi}$.
\end{lemma}

\begin{lemma}
    \label{Lem:Purificationsrelatedbyunitary}
    Let $\rho_A, \sigma_A$ be quantum states such that $\|\rho_A - \sigma_A\|_1 \le \epsilon$ and let $\ket{\phi_\rho}_{AB}, \ket{\phi_\sigma}_{AB}$ be purifications of $\rho, \sigma$. Then, there exists a unitary $U_B$ acting only on the purifying register $B$ such that $F\left(\ketbrab{\phi_\sigma}, (I_A \otimes U_B)\ketbrab{\phi_\rho} (I_A \otimes U_B) \right) \ge (1 - \epsilon/2)^2$.
\end{lemma}
\begin{proof}
    By the Fuchs-van de Graaf Inequality, $F(\rho, \sigma) \ge (1 - \epsilon/2)^2$. By Uhlmann's Theorem there are purifications $\ket{\psi_\rho}$ and $\ket{\psi_\sigma}$ with $|\braket{\psi_\sigma}{\psi_\rho}|^2 \ge (1 - \epsilon/2)^2$. By \cref{Lem:freedomofpurifications}, there are unitaries $U_\rho$ and $U_\phi$ such that $(I\otimes U_\rho)\ket{\phi_\rho} = \ket{\psi_\rho}$ and $(I\otimes U_\sigma)\ket{\phi_\sigma} = \ket{\psi_\sigma}$. It is now clear that $U_B \coloneq U_\sigma^\dagger U_\rho$ has the claimed property.
\end{proof}

\begin{lemma}[{Fannes inequality (see e.g. \cite[Theorem~11.6]{NC10})}]
    \label{Lem:Fannes}
    Let $\rho$ and $\sigma$ be two quantum states on a $d$-dimensional Hilbert space and let $T = \|\rho - \sigma\|_1$. Then,
    \begin{equation}
        |S(\rho) - S(\sigma)| \le T\log d - T \log T,
    \end{equation}
    when $T \le \frac{1}{e}$ and
    \begin{equation}
        |S(\rho) - S(\sigma)| \le T\log d + \frac{1}{e \ln 2}
    \end{equation}
\end{lemma}

\begin{lemma}[{\cite[Theorem~11.10]{NC10}}]
    \label{Lem:HolevoEntropyUpperBound}
    Let $\rho = \sum_{x} p_x \rho_x$ be a quantum state, where the $p_x$ form a probability distribution. Then,
    \begin{equation}
        S(\rho) \le \sum_{x} p_x S(\rho_x) + H(p),
    \end{equation}
    where $H(p)$ is the Shannon entropy of $p$. 
\end{lemma}

\begin{lemma}[{Flattening Lemma (see e.g. \cite[Lemma~5.6]{BST08})}]
    \label{Lem:Flattening}
    Let $\rho$ be an $n$-qubit quantum state. Then, for all $\epsilon >0$, there is a state $\sigma$ on $qn$ qubits such that $\|\rho^{\otimes q} - \sigma\|_1 \le 2 \epsilon$ and
    \begin{equation}
        H_{\infty}(\sigma) \ge q S(\rho) - O\left( n + \log \frac{q}{\epsilon} \right) \sqrt{q \log \frac{1}{\epsilon}}.
    \end{equation}
\end{lemma}


\section{Main technical lemmas}\label{scn:lemmas}
\subsection{Entropy verification in $\qqQAM$}
Here we show how to verify entropy in $\qqQAM$. This approach has previously appeared in \cite[proof of Lemma~33]{KGN19} and is based on ideas from \cite[Section~5]{BST08}. We explicitly state the procedure as a lemma for completeness and to hopefully ease further use.

Let $A,B$ be registers of $n_A, n_B$ qubits respectively, where $n_A + n_B = n$, and let $\tau = \tau(n)$ be the entropy threshold.  Consider the registers $\calA = A_1 \dots A_q$ and $\calB = B_1 \dots B_q$ consisting of $q$ copies of the $A$ and $B$ registers. Fix an explicit $(k, d, \epsilon)$-extractor $T \colon L(\calH_\calA) \to L(\calH_\calA)$ with $k = q\tau - O\left(n_A + \log \frac{q}{\epsilon}\right) \sqrt{q \log \frac{1}{\epsilon}}$, $d = qn_A - k + 2\log \frac{qn_A}{\epsilon} + O(1)$ and $q \le \poly(n)$ and $\epsilon$ to be determined later. Such an extractor exists by \cref{Lem:ExpandersExist}. Let $U_T$ be an efficient unitary implementing $T$ acting on $\calA$ and a $2qn_A$ qubit register $E$. That is, let $U_T$ be such that 
    \begin{align}
        T(\rho) = \Tr_E(U_T (\rho \otimes \ketbrab{0^{2qn_A}}) U_T^\dagger).
    \end{align}
    Now, consider the protocol, depending on $\tau$, $q$ and $\epsilon$, that on input $\chi_{\calA \calB E}$ does the following:
    \begin{enumerate}
        \item Apply $U_T^\dagger \otimes I_\calB$ to the input state.
        \item Perform a projective measurement $\{\Pi, I - \Pi\}$ where $\Pi = I_\calA \otimes I_\calB \otimes \ketbrab{0^{2qn_A}}_E$. Reject if the outcome is not $\Pi$.
        \item Accept and output the resulting state $\sigma_{\calA \calB}$.
    \end{enumerate}
    Note that the $\calB$ register is not acted on at all. Its sole purpose is to allow for correlations with a system whose entropy is not checked.

To verify entropy in $\qqQAM$ the verifier can do the following:
\begin{enumerate}
    \item Prepare EPR pairs across the $\calA$ and an $\calA'$ register.
    \item Send the $\calA'$ register to the prover.
    \item Receive two proof registers from the prover: $\calB$ and $E$. The total proof state will now be $\chi_{\calA\calB E}$. Note that $\chi_{\calA} = \tilI_\calA$ as the prover did not touch this register.
    \item Run the above protocol on the entire state $\chi_{\calA\calB E}$. If the protocol rejects, output NO.
    \item Run additional verification on the states $\sigma_{A_i B_i}$ as necessary. Accept iff these verifications accept for sufficiently many $i$ (depending on the exact problem).
\end{enumerate}
Doing so is supposed to have two effects: (1) For any high entropy state $\rho_{AB}$, a cooperating prover can make the output $\sigma_{\calA\calB}$ be (close to) the $q$-fold tensor product of $\rho_{AB}$. (2) If the prover tries to cheat and make the average entropy of the $\sigma_{A_i B_i}$ low, then the success probability of the protocol will be small. After running the protocol, the verifier can thus be confident that the $\sigma_{A_i B_i}$ are high entropy on average.

We now formally prove this:
\begin{lemma}[Entropy verification in $\qqQAM$]
    \label{Lem:EntropyVerification}
    For all $\tau$ and all $\delta, \delta' \ge \frac{1}{\poly(n)}$ there are completeness and soundness parameters $\ccomp, \csound$ with $\ccomp = 1 - 2^{-\Theta(n)}$ and $\ccomp - \csound \ge \frac{1}{\poly(n)}$ and $\epsilon, q$ with $\epsilon \ge 2^{-\Theta(n)}$ and $q \le \poly(n)$ such that the entropy verification protocol above satisfies the following properties:
    \begin{enumerate}
        \item For any state $\rho_{AB}$ with $S(\rho_A) \ge \tau$ there exists some state $\chi_{\calA \calB E}$ with $\chi_\calA = \tilI_\calA$ that is accepted by the protocol with probability at least $\ccomp$. The output state $\sigma_{\calA \calB}$ in case of acceptance satisfies
        \begin{align}
            \left\| \rho_{AB}^{\otimes q} - \sigma_{\calA \calB} \right\|_1 \le \delta.
        \end{align}
        Furthermore, the average output state $\tilde{\sigma}_{AB} \coloneq \frac{1}{q}\sum \sigma_{A_i B_i}$ has $\|\rho_{AB} - \tilde{\sigma}_{AB}\|_1 \le \delta$.

        \item Let $\chi_{\calA \calB E}$ be an arbitrary state satisfying $\|\chi_\calA - \tilI_\calA\| \le \der$ with $\der \le \frac{\delta'}{4n_A}$ and suppose that on input $\chi$, the measurement accepts with probability at least $\csound$. Then, the output state in case of acceptance $\sigma_{\calA,\calB}$ satisfies $S(\sigma_\calA) \ge q(\tau - \delta')$ and the average output state $\tilde{\sigma}_{AB}$ has $S(\tilde{\sigma}_A) \ge \tau - \delta'$.
    \end{enumerate}    
\end{lemma}
\begin{proof}
    For the first property, let $\rho_{AB}$ satisfy $S(\rho_A) \ge \tau$ and consider 
    \begin{align}
        \chi_{\calA \calB E} = (U_T \otimes I_\calB) \left( \rho_{A_1 B_1} \otimes \dots \otimes \rho_{A_q B_q} \otimes \ketbrab{0^{2qn_A}}_E \right) (U_T^\dagger \otimes I_\calB).
    \end{align}
    Note that by definition of $U_T$, $\Tr_{\calB E}(\chi) = T(\rho_A^{\otimes q})$. As the entropy of $\rho_A$ is high by assumption, the Flattening Lemma guarantees the existence of a state $\eta_\calA$ with $\|\rho_A^{\otimes q} - \eta_\calA \|_1 \le 2\epsilon$ and 
    \begin{equation}
        H_\infty(\eta_\calA) \ge q\tau - O\left(n_A + \log \frac{q}{\epsilon}\right) \sqrt{q \log \frac{1}{\epsilon}} = k. 
    \end{equation}
    By the second property of \cref{Def:Extractor} we have:
    \begin{align}
        \|\Tr_{\calB E}(\chi) - \tilI_\calA\|_1 &= \|T(\rho_A^{\otimes q}) - \tilI_\calA\|_1 \\ 
        &\le \|T(\rho^{\otimes q}) - T(\eta_\calA)\|_1 + \|T(\eta_\calA) - \tilI_\calA\|_1 \\
        &\le \|\rho^{\otimes q} - \eta_\calA\|_1 + \|T(\eta_\calA) - \tilI_\calA\|_1 \\
        &\le 2\epsilon + \epsilon.
    \end{align}
    
    Consider a purification $\ket{\chi}_{\calA\calB E P}$ of $\chi_{\calA \calB E}$ (and hence of $\chi_\calA)$, where $P$ is some purifying register. As $\|\chi_\calA - \tilI_\calA\|_1 \le 3\epsilon$, \cref{Lem:Purificationsrelatedbyunitary} guarantees the existence of a state $\ket{\tilde{\chi}}_{\calA \calB E P}$ with $\tilde{\chi}_{\calA} = \tilI_\calA$ and 
    \begin{align}
        F(\chi_{\calA\calB E}, \tilde{\chi}_{\calA\calB E}) \ge F(\ketbrab{\chi}_{\calA\calB E P}, \ketbrab{\tilde{\chi}}_{\calA\calB E P}) \ge (1 - 3\epsilon/2)^2 \ge 1 - 3\epsilon.
    \end{align}

    Using that fidelity is invariant under unitaries and the Fuchs-van de Graaf Inequality, we see that (dropping the identity applied to the $\calB$ register for legibility)
    \begin{align}
        \begin{split}
            \| U_T^\dagger \tilde{\chi}_{\calA \calB E} U_T - \rho_{AB}^{\otimes q} \otimes \ketbrab{0^{2qn_A}}_E \|_1 &\le \| U_T^\dagger \tilde{\chi} U_T - U_T^\dagger \chi U_T\|_1 \\ &\qquad+ \|U_T^\dagger \chi U_T - \rho^{\otimes q} \otimes \ketbrab{0^{2qn_A}} \|_1 
        \end{split} \\
        &= \|\tilde{\chi} - \chi\|_1 + 0 \\
        &\le 2\sqrt{3\epsilon} = \sqrt{12\epsilon}.
    \end{align}
    It follows that the measurement in step 2) of the protocol accepts with probability at least $\ccomp \coloneq 1 - 2\sqrt{3 \epsilon}$. By the Gentle Measurement Lemma, the post-measurement state 
    \begin{align}
        \sigma_{\calA \calB E} = \frac{\Pi U_T^\dagger \tilde{\chi}_{\calA \calB E} U_T \Pi}{\Tr\left(\Pi U_T^\dagger \tilde{\chi}_{\calA \calB E} U_T\right)}
    \end{align}
    satisfies
    \begin{align}
        \|\sigma_{\calA \calB} - \rho_{AB}^{\otimes q}\|_1 \le 2 \sqrt[4]{12 \epsilon}, 
    \end{align}
    which immediately implies $\|\tilde{\sigma}_{AB} - \rho_{AB}\|_1 \le 2\sqrt[4]{12 \epsilon}$.
    Choosing $\epsilon$ such that $2 \sqrt[4]{12 \epsilon} \le \delta$ we see that $\tilde{\chi}$ is as claimed.

    We now turn to the second property. Let $\chi_{\calA \calB E}$ be a state with $\|\chi_\calA - \tilI_\calA\|_1 \le \der$ and assume that 
    \begin{align}
        \Tr\left( \Pi \left(U_T^\dagger \otimes I_\calB\right) \chi_{\calA \calB E} \left(U_T \otimes I_\calB \right)\right) \ge \csound
    \end{align}
    By the Gentle Measurement Lemma, $U_T^\dagger \chi U_T $ is close to the post-measurement state $\sigma_{\calA \calB}$:
    \begin{align}
        \left\|(U_T^\dagger \otimes I_\calB) \chi_{\calA \calB E} (U_T\otimes I_B) - \sigma_{\calA \calB} \otimes \ketbrab{0^{2qn_A}}_E\right\|_1 \le 2\sqrt{1 - \csound}.
    \end{align}

    Now, consider $T(\sigma) = \Tr_{\calB E}(U_T (\sigma \otimes \ketbrab{0^{2qn_A}})U_T^\dagger)$ (omitting the $I_\calB$). Using that trace distance is non-increasing under partial trace and invariant under unitaries, we have:
    \begin{align}
        \|T(\sigma) - \tilI_\calA\|_1 &\le \| T(\sigma) - \chi_\calA\|_1 + \|\chi_\calA - \tilI_\calA\|_1 \\
        &\le \| U_T (\sigma_\calA \otimes \ketbrab{0^{2qn_A}}_E) U_T^\dagger - \chi_{\calA E}\|_1 + \der \\
        &\le \| \sigma_\calA \otimes \ketbrab{0^{2qn_A}}_E - U_T^\dagger \chi_{\calA E} U_T \|_1 + \der \\
        &\le 2 \sqrt{1 - \csound} + \der.
    \end{align}
    By the Fannes Inequality:
    \begin{align}
        S(T(\sigma_\calA)) \ge qn_A - \left(2\sqrt{1 - \csound} + \der\right)\left( qn_A - \log\left(2 \sqrt{1 - \csound} + \der\right)\right).
    \end{align}
    Because $T$ is $2^d$-regular, it can increase the entropy by at most $d$. It follows that:
    \begin{align}
        S(\sigma_\calA) \ge qn_A - \left(2\sqrt{1 - \csound} + \der\right)\left( qn_A - \log\left(2 \sqrt{1 - \csound} + \der\right)\right) - d.
    \end{align}
    Recall that $k = q\tau - O\left( n_A + \log \frac{q}{\epsilon} \right) \sqrt{q \log \frac{1}{\epsilon}}$ and that $d = qn_A - k + 2 \log \frac{qn_A}{\epsilon} + O(1)$. Therefore $d = q(n_A - \tau) + O\left( n_A + \log \frac{q}{\epsilon} \right) \sqrt{q \log \frac{1}{\epsilon}}$ and we have:
    \begin{align}
        S(\sigma) \ge q\left( \tau - \left(2\sqrt{1 - \csound} + \der\right)\left(n_A - \frac{1}{q}\log\left(2 \sqrt{1 - \csound} + \der\right)\right) - O(n_A + \log \frac{q}{\epsilon})\sqrt{\frac{1}{q}\log \frac{1}{\epsilon}} \right).
    \end{align}
    We would like $S(\sigma) \ge q(\tau - \delta')$, which requires
    \begin{equation}
        \label{Eq:DeltaRequirement}
        \left(2\sqrt{1 - \csound} + \der\right)\left(n_A - \frac{1}{q}\log\left(2 \sqrt{1 - \csound} + \der\right)\right) + O(n_A + \log \frac{q}{\epsilon})\sqrt{\frac{1}{q}\log \frac{1}{\epsilon}} \le \delta'.
    \end{equation}
    We now need to set $\epsilon, q$ and $\csound$ such that $\ccomp = 1 - 2^{-\Theta(n)}$, $\ccomp - \csound \ge \frac{1}{\poly(n)}$, $ \sqrt[4]{12 \epsilon} \le \delta$ and \cref{Eq:DeltaRequirement} are satisfied (recall that we previously set $\ccomp = 1 - 2\sqrt{3\epsilon})$.

    This can indeed be achieved. By assumption on $\der$, $n_A \der \le \frac{\delta'}{4}$. Setting $\epsilon = 2^{-\Theta(n)}$, makes $\ccomp = 1 - 2^{-\Theta(n)}$ and $\sqrt[4]{12\epsilon} \le \delta$. Next, we set $\csound$ such that $2 \sqrt{1 - \csound} n_A \le \frac{\delta'}{4}$. Note that an $\csound$ with $1 - \csound \ge \frac{1}{\poly(n)}$ suffices, which immediately yields $\ccomp - \csound  \ge \frac{1}{\poly(n)}$. Finally, we set $q$ to a large enough polynomial in $n$ such that \cref{Eq:DeltaRequirement} holds. Note that even though $\epsilon$ is exponentially small, a polynomial $q$ suffices as $\epsilon$ only appears under a logarithm.
    
    Now, consider the average output state $\tilde{\sigma}_{AB} = \frac{1}{q}\sum_{i = 1}^q \sigma_{A_i B_i}$. By concavity and subadditivity of von Neumann entropy we have 
    \begin{align}
        S(\tilde{\sigma}_A) &= S\left(\frac{1}{q} \sum_{i = 1}^q \sigma_{A_i} \right) \\
        &\ge \frac{1}{q} \sum_{i = 1}^q S(\sigma_{A_i}) \\
        &\ge \frac{1}{q} S(\sigma_{\calA}) \\
        &\ge \tau - \delta',
    \end{align}
    completing the proof. 
\end{proof}


\subsection{Channel-to-Hamiltonian construction}
The main tool for our hardness proofs is a construction that allows embedding the range of a quantum channel in the ground space of local Hamiltonian.
\begin{lemma}[Channel-to-Hamiltonian construction]
    \label{Lem:ChanneltoHamiltonian}
    Let $\Phi\colon L(\calH_\Rin) \to L(\calH_\Rout)$ be a quantum channel and let $\nin, \nout$ be the number of qubits of $\calH_\Rin$ and $\calH_\Rout$ respectively. Let $E$ be a register of $\nin + \nout$ qubits and consider a unitary $U_\Phi = V_T \dots V_1$ acting on $\Rin\Rout E$ implementing the Stinespring representation of $\Phi$:
    \begin{align}
        \Phi(\rho) = \Tr_{AE}\left( U_\Phi (\rho_A \otimes \ketbrab{0^{n_B}}_B \otimes \ketbrab{0^{n_A + n_B}}_E) U_\Phi^\dagger \right).
    \end{align}
    Then, there exists some $5$-local Hamiltonian $H_\Phi$, depending on a parameter $L$ and acting on the registers $\Rin,\Rout,E$ and a $T + L$-qubit clock register $C$ such that
    \begin{enumerate}
        \item For any $\sigma_A$, there exists some $\rho_{\Rin\Rout CE}$ (pure if $\sigma$ is pure) such that $\Tr(H \rho) = 0$ and $\|\rho_\Rout - \Phi(\sigma)\|_1 \le \frac{2T}{T + L + 1}$.
        \item If $\rho_{\Rin \Rout CE}$ is such that $\Tr(H \rho) \le b$, then there exists some state $\sigma_\Rin$ (pure if $\rho$ is pure) such that $\|\rho_\Rout - \Phi(\sigma)\|_1 \le 2 \sqrt{\frac{b}{\Delta}} + \frac{2T}{T + L + 1}$. Here $\Delta$ is the spectral gap of $H_\Phi$, which satisfies $\Delta = \Omega\left( (T + L + 1)^{-3}\right)$.
    \end{enumerate}
\end{lemma}
\begin{proof}
    We will use Kitaev's circuit-to-Hamiltonian construction applied to the unitary $U_\Phi$ with two modifications. We use no output penalty and post idle for $L$ time steps. The Hamiltonian $H_\Phi$ will be given by
    \begin{align}
        H_\Phi &= H_{\textup{in}} + H_{\textup{prop}} + H_{\textup{clock}} \\
        H_{\textup{in}} &= \left(\sum_{i = 1}^{\nout} \ketbrab{1}_{\Rout_{,i}}  + \sum_{j = 1}^{\nin + \nout} \ketbrab{1}_{E_j}\right) \otimes \ketbrab{0}_C \\ 
        H_{\textup{prop}} &= \sum_{t = 0}^{T + L -1} -V_{t+1} \otimes \ketbra{t+1}{t}_C - V_{t+1}^{\dagger} \otimes \ketbra{t}{t+1}_C + I \otimes \ketbrab{t}_C + I \otimes \ketbrab{t+1}_C \\
        H_{\textup{clock}} &= I_{\Rin\Rout E} \otimes \sum_{i = 1}^{T + L - 1} \ketbrab{01}_{C_iC_{i+1}},
    \end{align}
    where $V_t = I$ for $t > T$.

    The ground state of $H_\Phi$ will now consist of \emph{history states} of the form:
    \begin{align}
        \ket{\psi_{\textup{hist}}} &= \frac{1}{\sqrt{T + L + 1}}\sum_{t = 0}^{T+L} V_t \dots V_1 \ket{\psi}_\Rin \ket{0^{\nout}}_\Rout\ket{0^{\nin + \nout}}_E \ket{t}_C\\
        \label{Eq:HistoryState}
        \begin{split}
            &= \frac{1}{\sqrt{T + L +1}}\sum_{t = 0}^{T-1} V_t \dots V_1 \ket{\psi}_\Rin \ket{0^\nout}_\Rout\ket{0^{\nin + \nout}}_E \ket{t}_C\\ &\qquad+ U_\Phi\ket{\psi}_\Rin\ket{0^\nout}_\Rout\ket{0^{\nin + \nout}}_E \left(\frac{1}{\sqrt{T + L + 1}} \sum_{t = T}^{T+L}\ket{t}_C\right),
        \end{split}
    \end{align}
    where $\ket{t}_C = \ket{1^t0^{T + L - t}}$ encodes the time step $t$ in unary.
    As we use no output penalty, any history state will have zero energy and all zero-energy states will be history states. It is shown in \cite[Lemma~23]{GK12} that Kitaev Hamiltonians without output penalty, and in particular $H_\Phi$, have spectral gap $\Delta \ge \Omega((T + L + 1)^{-3})$.\footnote{Note that our use of $\Delta$ differs from that of \cite{GK12}. We use it for the spectral gap of $H_\Phi$ whereas \cite{GK12} use it for a scaling constant.}

    To prove property 1, we start by considering pure states. Let $\ket{\psi}_\Rin$ be an arbitrary pure state. The reduced state on the $\Rout$ register of $\ket{\psi_{hist}}$ the history state on input $\ket{\psi}_\Rin$ is given by\footnote{Note that there are no cross-terms as all terms in the sum are orthogonal in the $C$ register.}
    \begin{align}
        \eta_\psi \coloneq \Tr_{\Rin CE}\left( \ketbrab{\psi_{hist}}_{\Rin \Rout CE}\right) = \frac{L + 1}{T + L + 1} \Phi(\psi) + \frac{1}{T + L + 1} \sum_{t = 0}^{T - 1} \rho'_{t, \psi},
    \end{align}
    where $\rho'_{t, \psi} = \Tr_{\Rin CE} \left(V_t \dots V_1 \ket{\psi}_A\ket{0^{\nout}}_\Rout \ket{0^{\nin + \nout}}_E \ket{t}_C \right)$ are error terms.

    Now consider an arbitrary mixed state $\rho_\Rin = \sum_x a_x \ketbrab{\psi_x}_\Rin$. Write $\ket{\psi_{hist,x}}$ for the history state on input $\ket{\psi_x}$. By linearity of $\Phi$ and the partial trace we have:
    \begin{align}
        \eta_\rho &\coloneq \Tr_{\Rin CE}\left( \sum_x a_x \ketbrab{\psi_{hist,x}}\right) \\
        &= \frac{L + 1}{T + L + 1} \sum_x a_x \Phi(\psi_x) + \frac{1}{T + L + 1} \sum_x a_x \sum_{t = 0}^{T - 1} \rho'_{t, x} \\
        &= \frac{L + 1}{T + L + 1} \Phi(\rho) + \frac{1}{T + L + 1} \sum_{t = 0}^{T - 1} \rho''_{t}. \label{Eq:TraceMixtureHistoryState}
    \end{align}
    Here the $\rho''_{t} = \sum_x a_x \rho'_{t, x}$ are again error terms. The trace distance between $\eta_\rho$ and $\Phi(\rho)$ can be upper bounded as
    \begin{align}
        \|\Phi(\rho) - \eta_\rho\|_1 &= \left\| \frac{T}{T + L + 1} \Phi(\rho) - \frac{1}{T + L + 1} \sum_{t = 0}^{T - 1} \rho''_{t} \right\|_1        \\
        &= \frac{T}{T + L + 1} \left\| \Phi(\rho) - \frac{1}{T}\sum_{t = 0}^{T - 1} \rho''_t \right\|_1 \\
        &\le \frac{2T}{T + L + 1}.
    \end{align}

    To prove the second property, let $\rho_{\Rin \Rout CE}$ be some state with $\Tr(H \rho) \le b$ and let $\Pi$ be the projector on the span of the history states, or equivalently, on the zero-energy space of $H_\Phi$. Note that $b \ge \Tr(H \rho) \ge \Delta\Tr((I - \Pi)\rho) = \Delta\left(1 - \Tr(\Pi \rho)\right)$. Hence $\Tr(\Pi \rho) \ge 1 - \frac{b}{\Delta}$. By the Gentle Measurement Lemma, the state $\eta_{\Rin \Rout CE} = \frac{\Pi \rho \Pi}{\Tr(\Pi \rho)}$ now satisfies $\|\rho - \eta\|_1 \le 2 \sqrt{\frac{b}{\Delta}}$. Furthermore, $\eta$ is a mixture of history states. That is, there are states $\ket{\psi_x}_\Rin$ such that $\eta = \sum_x a_x \ketbrab{\psi_{hist, x}}$ where $\ket{\psi_{hist, x}}$ is the history state on input $\ket{\psi_x}$. Note that if $\rho$ is pure, then $\eta$ is pure and a sum over a single item suffices. As in \cref{Eq:TraceMixtureHistoryState} we have 
    \begin{align}
        \eta_{\Rout} = \frac{L + 1}{T + L + 1}\Phi(\sigma) + \frac{1}{T + L + 1}\sum_{t = 0}^{T - 1} \sigma'_t,
    \end{align}
    where $\sigma = \sum_x a_x \ketbrab{\psi_x}$ is in particular pure if $\rho$ is pure, and the $\sigma'_t$ are error terms. The distance between $\rho_\Rout$ and $\Phi(\sigma)$ is given by:
    \begin{align}
        \|\rho_\Rout - \Phi(\sigma)\|_1 &\le \|\rho_\Rout - \eta_\Rout\|_1 + \|\eta_\Rout - \Phi(\sigma)\|_1 \\
        &\le 2 \sqrt{\frac{b}{\Delta}} + \frac{2T}{T + L + 1},
    \end{align}
    completing the proof.
\end{proof}


\section{The High Entropy Ground State problem}\label{scn:heles}
Recall the definition of $\HELES$. 
\defHELES*
We now show that $\HELES$ is complete for $\qqQAM$ as claimed in \cref{thm:heles}
Treatment of physical Hamiltonians is postponed to \cref{sec:physicalHamiltonians}. We begin by showing containment.


\begin{lemma}\label{Lem:HELEScontainment}
    For all constant $k$, $\HELES \in \qqQAM$. 
\end{lemma}
\begin{proof}
    We will show the slightly stronger statement that $\HELES \in \qqQAM$, even if we replace the pure state in the YES case of the definition of $\HELES$ with a mixed state.\footnote{Note that the YES case of the pure state variant is a subset of the YES case of the mixed state variant. Hence the mixed state variant is at least as hard as the pure state variant. In fact, our results show that they are equally hard.}

    Let $(H_{AB}, \El, b, \fraks, \frakt)$ be a $\HELES$ instance. We use the entropy verification protocol from \cref{Lem:EntropyVerification} to ensure that the prover sends a state with high entropy on the $A$ register. Let $q, \calA, \calB$ and $E$ be as in \cref{Lem:EntropyVerification}. We split the $E = E_LE_R$ register in 2 halves of $qn_A$ qubits each.
    The $\calA$ register will contain the verifier's halves of the EPR states, and the prover's halves will be in the $E_L$ register. The prover will act on the $\calB$ and $E$ registers and send them to the verifier who then has access to the entire proof state $\chi_{\calA \calB E}$. 

    To measure the energy, we will need to use parallel repetition. We thus consider $N$ copies $\calA_1, \dots, \calA_N, \calB_1, \dots, \calB_N, E_1, \dots, E_N$ of all the registers. We will write the total proof state as $\chi^{(1\dots N)}_{\calA_1, \calB_1, E_1, \dots, \calA_N, \calB_N, E_N}$ (sometimes omitting the registers for brevity), where $\chi^{(i)}_{\calA_i,\calB_i,E_i}$ denotes the reduced state on the $i$-th copies of the registers, which is the state we will use in the $i$-th repetition. We would like $\chi^{(1\dots N)}$ to be neatly in product, i.e. $\chi^{(1\dots N)} = \chi^{(1)} \otimes \dots \otimes \chi^{(N)}$, but in the NO case, there can be entanglement between the states used in the different repetitions. 

    The protocol for verifying $\HELES$ in $\qqQAM$ looks like this:
    \begin{enumerate}
        \item For all $i$, the verifier prepares EPR pairs across the $\calA_i$ and $E_{i,L}$ (the $i$-th copy of the left half of the $E$ register) registers. They then send all $E_{i,L}$ registers to the prover.
        \item The prover now acts on all $\calB_i$ and $E_i$ registers, and potentially their own private register $\calP$. They send the $\calB_i$ and $E_i$ registers to the verifier, who then has access to the entire proof state $\chi^{(1\dots N)}_{\calA_1, \calB_1, E_1, \dots, \calA_N, \calB_N, E_N}$, which is potentially entangled with the $\calP$ register. Note that the prover could not act on the $\calA_i$ registers, so the verifier is sure that $\chi_{\calA_1,\dots,\calA_N} = \tilI_{\calA_1, \dots \calA_N} = \tilI_{\calA_N} \otimes \dots \otimes \tilI_{\calA_N}$.
        \item For every set of registers $\calA_k\calB_k E_k$, the verifier does the following. First they run the entropy verification test with $\tau = \fraks, \delta = \frac{b - \El}{4m}$ and $\delta' < \fraks - \frakt$. If the test passes, they pick a uniformly random $A_i B_i$ from within these registers and a uniformly random local term $H_j$. They then measure the energy of the $\sigma^{(k)}_{A_iB_i}$ with respect to the Hamiltonian $H_j$. Let the random variable $\calE_k$ denote the outcome of this measurement. In the case that the entropy verification fails, the verifier wants to reject immediately. We simulate this by setting $\calE_k = \frac{b}{m(1-\csound)}$ and will see later that this indeed results in rejection.
        \item The verifier accepts if $\frac{1}{N}\sum_{k = 1}^{N} \calE_k < \frac{\El + b}{2m}$ and rejects otherwise.
    \end{enumerate} 
    
    We will now prove that in the YES case this protocol accepts with probability $1 - 2^{-\Omega(n)}$, whereas in the NO case it accepts with probability at most $1 - \frac{1}{\poly}$.

    \emph{Completeness.} Let $\rho_{AB}$ be such that $\Tr(H\rho) \le \El$ and $S(\rho_A) \ge \fraks$ and let $\hat{\chi}_{\calA \calB E}$ be such that the entropy verification protocol accepts with probability at least $\ccomp$, in which case the output state $\sigma_{\calA \calB}$ satisfies $\|\sigma_{\calA \calB} - \rho_{AB}^{\otimes q}\|_1 \le \delta$ and $\tilde{\sigma}_{AB} - \rho\|_1 \le \delta$. This is possible by property 1 of \cref{Lem:EntropyVerification}. The prover sends $N$ copies of the state $\hat{\chi}_{\calA \calB E}$, that is, the total proof state $\chi^{(1\dots N)}_{\calA_1, \calB_1, E_1, \dots, \calA_N, \calB_N, E_N} = \hat{\chi}_{\calA_1\calB_1E_1} \otimes \dots \otimes \hat{\chi}_{\calA_N\calB_N E_N}$.

    By a union bound, the probability that one of the runs of the entropy verification protocol fails is inverse exponential at worst (recall that $\ccomp = 1 - 2^{-\Theta(n)}$). Assuming this does not happen we have
    \begin{align}
        \EE(\calE_k) &= \frac{1}{mq}\sum_{i = 1}^{q} \sum_{j = 1}^{m} \Tr(H_j \sigma^{(k)}_{A_iB_i}) \\
        &= \frac{1}{m} \Tr\left(\sum_{j = 1}^m H_j \frac{1}{q}\sum_{i = 1}^{q} \sigma^{(k)}_{A_iB_i}\right) \\
        &= \frac{1}{m} \Tr(H \tilde{\sigma}^{(k)}_{AB}) \\
        &= \frac{1}{m}\Tr(H \rho) + \frac{1}{m}\Tr(H (\tilde{\sigma} - \rho)) \\
        &\le \frac{1}{m}\El + \frac{1}{m}\|H\|_\infty\|\tilde{\sigma} - \rho\|_1 \\
        &\le \frac{1}{m}\left(\El + \frac{b - \El}{4}\right). \label{eq:HolderApplication}
    \end{align} 
    Here we used the H\"older inequality and the fact that $m \ge \|H\|_\infty$. For completeness, we can assume that the prover cooperates and hence that the random variables $\calE_k$ are independent. Using a Hoeffding bound, we now get
    \begin{align}
        \Pr\left[ \frac{1}{N}\sum_{k = 1}^N \calE_k \ge \frac{\El + b}{2m}\right] & \le \exp\left(-\left(\frac{b - \El}{4m}\right)^2 N\right).
    \end{align}
    Setting $N$ to be a sufficiently large polynomial in $n$, the acceptance probability in the YES case thus becomes exponentially close to $1$.

    \emph{Soundness.} Let $\Ent_k$ denote the event that $S(\tilde{\sigma}^{(k)}_{A}) \ge \frakt$, and let $\acc_k$ and $\rej_k$ be the events where the entropy verification on input $\chi^{(k)}_{\calA_k\calB_kE_k}$ accepts and rejects, respectively. Consider now the expected value of $\calE_k$:
    \begin{align}
        \EE(\calE_k) &= \Pr(\Ent_k)\EE(\calE_k|\Ent_k) + \Pr(\neg \Ent_k)\EE(\calE_k|\neg\Ent_k)\\
        \EE(\calE_k|\neg\Ent_k) &= \Pr(\acc_k|\neg\Ent_k)\EE(\calE_k|\neg\Ent_k \wedge \acc_k) + \Pr(\rej_k|\neg\Ent_k)\EE(\calE_k|\neg\Ent_k \wedge \rej_k).
    \end{align}
    Note that by \cref{Lem:EntropyVerification} Property 2, $\Pr(\rej_k|\neg \Ent_k) \ge 1-\csound$. Furthermore, as $H_j \succeq 0$ for all $j$, $\EE(\calE_k|\neg\Ent_k \wedge \acc_k) \ge 0$ and $\EE(\calE_k|\neg\Ent_k \wedge \rej_k) = \frac{b}{m(1 - \csound)}$. Together this means that $\EE(\calE_k|\neg\Ent_k) \ge \frac{b}{m}$.

    If $S(\tilde{\sigma}^{(k)}_{A}) \ge \frakt$, then by the conditions of the NO case, $\Tr(H\tilde{\sigma}^{(k)}_{AB}) \ge b$. This means that:
    \begin{align}
        \EE(\calE_k|\Ent_k) &= \frac{1}{mq}\sum_{i = 1}^{q} \sum_{j = 1}^{m} \Tr(H_j \sigma^{(k)}_{A_iB_i}) \\
        &= \frac{1}{m} \Tr\left(\sum_{j = 1}^m H_j \frac{1}{q}\sum_{i = 1}^{q} \sigma^{(k)}_{A_iB_i}\right) \\
        &= \frac{1}{m} \Tr(H \tilde{\sigma}^{(k)}_{AB}) \\
        &\ge \frac{b}{m}.
    \end{align}
    It follows that $\EE(\calE_k) \ge \frac{b}{m}$ when we are in the NO case. Let $M = \max\left\{1, \frac{b}{m(1 - \csound)}\right\}$ and consider now the random variable $Y = M - \frac{1}{N}\sum_{k = 1}^{N} \calE_k$. By construction, $Y \ge 0$. Furthermore, $\EE(Y) \le M - \frac{b}{m}$. Markov's inequality now yields
    \begin{align}
        \Pr\left(\frac{1}{N}\sum_{k = 1}^N \calE_k \le \frac{\El + b}{2m}\right) &= \Pr\left(Y \ge M - \frac{\El + b}{2m}\right) \\
        &\le \frac{M - \frac{b}{m}}{M - \frac{\El + b}{2m}} \\
        &\le 1 - \frac{b - \El}{2mM - b + (b - \El)} \\
        &\le 1 - \frac{1}{\poly}, \label{eq:markov}
    \end{align}
    where in the last step we used that $1 \le M \le \poly(n)$, $b - \El \ge \frac{1}{\poly}$ and that $b \le m$. 
    
    The acceptance probabilities in the YES and NO cases are thus inverse polynomially separated. (Exponentially close to 1 in the YES case and polynomially bounded away from 1 in the NO case.) This completes the proof.

\end{proof}


\begin{lemma}
    \label{Lem:HELEShard}
    For $k \ge 5$ and $n_B = \poly(n_A)$, $\HELES$ is $\qqQAM$-hard.
\end{lemma}
\begin{proof}
    We show a poly-time many-one reduction from $\MaxOutQEA$. Let $(\Phi,\tau)$ be a $\MaxOutQEA$ instance. Without loss of generality we can assume that for all $\sigma$, there is a pure state $\ket{\psi}$ with $\Phi(\sigma) = \Phi(\psi)$.\footnote{To see this, modify $\Phi\colon L(\calH_\Rin) \to L(\calH_\Rout)$ by adding another copy of the input register that is immediately traced out. Taking $\ket{\psi}$ to be some purification of $\sigma$ then suffices.} Let $H_\Phi$ be the Hamiltonian acting on $\Rin, \Rout, C, E$ resulting from \cref{Lem:ChanneltoHamiltonian}, with $L$ to be determined later. The output of the reduction will be $(H_\Phi, \El, b, s, t)$, where the cut is between the $\Rout \eqcolon A$ and $\Rin CE \eqcolon B$ registers. The parameters $\El, b, s$ and $t$ will be determined later.

    For completeness, suppose $(\Phi, \tau)$ is a YES instance. That is, suppose there is some $\sigma_\Rin$ with $S\left(\Phi(\sigma)\right) \ge \tau + 1$. Let $\ket{\psi}_\Rin$ be a pure state such that $\Phi(\psi) = \Phi(\sigma)$. By property 1 of \cref{Lem:ChanneltoHamiltonian} there exists a zero-energy pure state $\ket{\phi}_{\Rin \Rout CD} = \ket{\phi}_{AB}$ such that $\|\phi_B - \Phi(\psi)\| \le \frac{2T}{T + L + 1}$. The Fannes Inequality now gives
    \begin{align}
        S(\phi_B) &\ge S\left(\Phi(\psi)\right) - \frac{2T}{T + L + 1}\left( n_B - \log \frac{2T}{T + L + 1}\right)
    \end{align}
    Setting $\El = 0$, $s = \tau + \frac{3}{4}$ and choosing $L$ such that
    \begin{equation}
        \label{Eq:Lreq1}
        \frac{2T}{T + L + 1}\left( n_B - \log \frac{2T}{T + L + 1}\right) \le \frac{1}{4},
    \end{equation} 
    we see that $\ket{\phi}_{AB}$ satisfies $S(\phi_B) \ge s$ and $\braketb{\phi}{H} \le \El$ as desired.

    For soundness, suppose that for all $\sigma_A$, $S\left(\Phi(\sigma)\right) \le \tau$ and let $\rho_{\Rin \Rout CE} = \rho_{AB}$ be such that $\Tr(H \rho_{AB}) \le b$ ($b$ will be determined later). By property 2 of \cref{Lem:ChanneltoHamiltonian}, there is a state $\sigma_\Rin$ such that 
    \begin{align}
        D \coloneq \|\rho_B - \Phi(\sigma_\Rin)\|_1 \le 2 \sqrt{\frac{b}{\Delta}} + \frac{2T}{T + L + 1}.
    \end{align}
    From the Fannes Inequality it follows that 
    \begin{align}
        S(\rho_B) &\le S(\Phi(\sigma)) + D(n_B - \log D)\\
        &\le \tau + \left( \sqrt{\frac{b}{\Delta}} + \frac{2T}{T + L + 1}\right) \left(n_B - \log\left( \sqrt{\frac{b}{\Delta}} + \frac{2T}{T + L + 1}\right)\right).
    \end{align}
    Choosing $t = \tau + \frac{1}{4}$, $b = \frac{\Delta}{L} = \Omega(L^{-4})$ and then setting $L$ large enough so that 
    \begin{align}
        \label{Eq:Lreq2}
        \left( \sqrt{\frac{b}{\Delta}} + \frac{2T}{T + L + 1}\right) \left(n_B - \log\left( \sqrt{\frac{b}{\Delta}} + \frac{2T}{T + L + 1}\right)\right) \le \frac{1}{4}
    \end{align}
    we see that $S(\rho_B) \le t$ as desired. Note that \cref{Eq:Lreq1,Eq:Lreq2} can be satisfied by an $L$ that is polynomially large in $n = n_A + n_B$.

    Close inspection of \cite{KGN19} shows that $\qqQAM$-hard $\MaxOutQEA$ instances can be assumed to map $n$ qubits to $n$ qubits. It follows that the size of the $A$ register is polynomially larger than that of the $B$ register.
\end{proof}


\section{Free energy and Gibbs states}\label{scn:FEA}
\begin{definition}[Free energy]
    The free energy of a Hamiltonian at inverse temperature $\bt$ is defined as:
    \begin{equation}
        \label{eq:FEvariational}
        \calF(H) = \min_{\rho} f(\rho) \quad \text{where} \quad f(\rho) = \Tr(H \rho) - \frac{S(\rho)}{\bt}. 
    \end{equation}
    Equivalently, the free energy is also equal to
    \begin{equation}
        \calF(H) = -\frac{\log \calZ}{\bt}.
    \end{equation} 
    Here $\calZ = \Tr(e^{-\bt H})$ is also known as the \emph{partition function}.
\end{definition}

The state that minimizes the free energy is known as the \emph{Gibbs state} and is given by $\rho_\bt = \frac{1}{\calZ} e^{-\bt H}$.

\defFEA*

We now formally prove the $\qqQAM$ containment of $\FEA$.
\thmFEA*
\begin{proof}
    Let $(H, \bt, a, b)$ be an $\FEA$ instance, where $H = \sum_{j = 1}^m H_j$ is a $k$-local Hamiltonian on an $n$-qubit register $A$. 
    

    Let $\calA = A_1 \dots A_q$ and $E$ be as in \cref{Lem:EntropyVerification}. We do not use the $\calB$ register and split $E = E_LE_R$ into two registers of $qn$ qubits each.

    The verifier's actions will be similar as that in the proof of $\HELES \in \qqQAM$ (\Cref{Lem:HELEScontainment}). We again use parallel repetition for an energy measurement, so we consider $N$ copies $\calA_1, \dots \calA_N$ and $E_1 \dots E_N$ of the $\calA$ and $E$ registers. The protocol for verifying $\FEA$ consists of the following steps. 
    \begin{enumerate}
        \item Verifier prepares EPR pairs across all copies of the $\calA$ and $E_L$ registers. They then send the $E_L$ registers to the prover.
        \item The prover now acts on the $E$ registers, and potentially their own private register $\calP$. They send the $E$ registers to the verifier, who then has access to the entire proof state $\chi^{(1,\dots, N)}_{\calA_1 E_1 \dots \calA_N E_N}$. Furthermore, the prover sends a classical value $\tilde{S}$ to the verifier. 
        \item If $\tilde{S} \not\in [0,n]$ the verifier rejects. Otherwise, they do the following for every pair $\calA_k, E_k$. First, they run the entropy verification protocol with $\delta \le \frac{b - a}{4m}, \delta' < \bt \frac{b - a}{4m}$ but $\delta,\delta'\ge \frac{1}{\poly}$ and $\tau = \tilde{S}$. If the protocol passed, they pick a uniformly random local term $H_j$ and a subregister $A_i$ of $\calA_k$ and measure $\Tr(H_j \sigma^{(k)}_{A_i})$. Let the random variable $\calE_k$ denote the outcome of this measurement and let $F_k = \calE_k - \frac{\tilde{S}}{m\bt}$. If the entropy verification protocol fails, we set $F_k = \frac{1}{1-\csound}\left(\frac{|b|}{m} + \frac{\csound \tilde{S}}{m \bt}\right)$.
        \item The verifier now accepts if $\frac{1}{N}\sum_{k = 1}^{N} F_k < \frac{a + b}{2m}$ and rejects if it is bigger or equal to this value.
    \end{enumerate}
    
    For completeness, let $\rho_\bt = \frac{1}{\calZ}e^{-\bt H}$ be the Gibbs state at inverse temperature $\bt$. In the YES case this satisfies $f(\rho_\bt)\le a$. The prover sends the classical value $\tilde{S} = S(\rho_\bt)$ and the $E$ registers such that the proof states $\chi^{(k)}_{\calA_k E_k}$ are in product and that the entropy verification protocol on the $\chi^{(k)}$ succeed with probability at least $\ccomp$, in which case the average output state $\tilde{\sigma}^{(k)}_{\calA_k}$ satisfies $\|\rho_\bt - \tilde{\sigma}^{(k)}\|_1 \le \delta$. This is possible by property 1 of \cref{Lem:EntropyVerification}.

    Note that $\ccomp$ is exponentially close to 1, so by a union bound, the probability that one of the runs of the entropy verification protocol fails is exponentially small. Assuming this does not happen we can use H\"older's inequality in a similar way to \cref{eq:HolderApplication} to get 
    \begin{align}
        \EE(F_k) &\le \frac{1}{m}\left(\Tr(H\rho_\bt) + \frac{b-a}{4} - \frac{\tilde{S}}{\bt} \right) \\
        &\le \frac{1}{m}\calF(H) + \frac{b - a}{4m}.
    \end{align}
    As we are in the YES case, $\calF(H) \le a$. Using a Hoeffding bound we see that the probability of rejection in the YES case is exponentially small, for sufficiently large (but still polynomial) $N$.

    For soundness, assume that $\calF(H) \ge b$. Let $\chi^{(k)}_{\calA E}$ be the (reduced) proof state used in the $k$-th repetition of the protocol and $\tilde{S}$ be the prover's entropy claim. (Note that we cannot, and will not, assume that the $\chi^{(k)}_{\calA E}$ are uncorrelated with each other.) Let $\Ent_k$ denote the event where the output of the entropy verification (in case of acceptance) on input $\chi^{(k)}_{\calA E}$ has high entropy: $S(\tilde{\sigma}^{(k)}_A) \ge \tilde{S} - \delta'$. Furthermore, let $\acc_k$ and $\rej_k$ be the events where the entropy verification protocol on input $\chi^{(k)}_{\calA E}$ accepts and rejects, respectively. We have
    \begin{align}
        \EE(F_k|\neg\Ent_k) &= \Pr(\acc_k|\neg\Ent)\EE(F_k|\neg\Ent_k \wedge \acc_k) + \Pr(\rej|\neg\Ent_k)\EE(F_k|\neg\Ent_k \wedge \rej_k).
    \end{align}
    Note that $\Pr(\rej_k|\neg\Ent_k) \ge 1 - \csound$ and $\Pr(\acc_k|\neg\Ent) \le \csound$ by \cref{Lem:EntropyVerification} Property 2. Furthermore, 
    \begin{align}
        \EE(F_k|\neg\Ent_k \wedge \rej_k) = \frac{1}{1-\csound}\left(\frac{|b|}{m} + \frac{\csound \tilde{S}}{m \bt}\right)
    \end{align}
    by definition, and 
    \begin{align}
        \EE(F_k|\neg\Ent \wedge \acc_k) \ge - \frac{\tilde{S}}{m\bt}
    \end{align} since $0 \preceq H_j$. It follows that $\EE(F_k|\neg\Ent_k) \ge \frac{|b|}{m} \ge \frac{b}{m}$.

    If the event $\Ent$ occurs, i.e., if $S(\tilde{\sigma}^{(k)}_\calA) \ge \tilde{S} - \delta'$, then $b \le \calF(H) \le \Tr(H\tilde{\sigma}^{(k)}_A) - \frac{\tilde{S} - \delta'}{\bt}$, so $\Tr(H\tilde{\sigma}^{(k)}_A) \ge b + \frac{\tilde{S} - \delta'}{\bt} > b + \frac{\tilde{S}}{\bt} - \frac{b - a}{4}$. It follows that $\EE(F_k | \Ent) \ge \frac{b}{m} - \frac{b - a}{4m}$ and that $\EE(F_k) \ge \frac{b}{m} - \frac{b - a}{4m}$.

    Let $M = \max\left\{1, \frac{1}{1-\csound}\left(\frac{|b|}{m} + \frac{\csound \tilde{S}}{m \bt}\right)\right\}$. We have $-\frac{n}{m\bt} \le \frac{\tilde{S}}{m\bt} \le F_k \le M$ for all $k$. It follows that $\frac{1}{N} \sum_{k = 1}^{N} F_k$ is a non-negative random variable with expectation
    \begin{align}
        \EE\left[M - \frac{1}{N} \sum_{k = 1}^{N} F_k\right] \le M - \frac{b}{m} + \frac{b - a}{4m}.
    \end{align}
    Using Markov's Inequality we now get
    \begin{align}
        \Pr\left[ \frac{1}{N} \sum_{k = 1}^{N} F_k \le \frac{a + b}{2m} \right] &= \Pr\left[M - \frac{1}{N} \sum_{k = 1}^{N} F_k \ge M - \frac{a + b}{2m} \right] \\
        &\le \frac{M - \frac{b}{m} + \frac{b - a}{4m}}{M - \frac{a + b}{2m}} \\
        &\le 1 - \frac{b - a}{4mM - 2a - 2b} \\
        &\le 1 - \frac{1}{\poly(n)}.
    \end{align}
    Here we used that $1 \le M \le \poly(n)$ and the fact that we can WLOG assume $-\poly(n) \le -\frac{n}{\bt} \le a \le b - \frac{1}{\poly(n)} \le b \le m$.

\end{proof}


\section{Low Entropy Ground State Problem}\label{scn:leles}
Recall the definition of $\LELES$.
\defleles*
We now show that $\LELES$ is $\QMAt$-hard, as claimed in \cref{thm:leles}.
\begin{theorem}
    For $k\ge 5$, $\LELES$ is $\QMAt$-hard, even when $s - t = \Omega(n)$.
\end{theorem}
\begin{proof}
    We now prove the statement for general Hamiltonians. Treatment of physical Hamiltonians is postponed to \cref{sec:physicalHamiltonians}.

    We will reduce from the $\QMA(2)$-complete problem $\PPIO$ (\Cref{def:PPIO}). Let $(U, \dy, \dn)$ be a $\PPIO$ instance and assume that $\dy = 2^{-\Omega(n)}$ and $\dn = 2 - 2^{-\Omega(n)}$. We map this to a $\LELES$ instance $(H_U, \El, \Eh, s, t)$ where $H_U$ is the Hamiltonian acting on registers $L,R,C$ resulting from \cref{Lem:ChanneltoHamiltonian} applied to the channel $\Phi(\rho) \coloneq U(\rho_L \otimes \ketbrab{0^{n_R}}_R) U^\dagger$. The cut will be between the $L \eqcolon A$ and $R C \eqcolon B$ registers. Note that $H_\Phi$ is indeed $5$-local. The parameters $s$ and $t$ will be $\Omega(n)$ and $O(1)$ respectively and $\El, \Eh$ will be fixed at a later point in the proof.

    For completeness, let $\ket{\psi}_L$ be such that $\Phi(\psi)$ is $\dy$-close to $\ket{\phi_L}\otimes\ket{\phi_R}$. By property 1 of \cref{Lem:ChanneltoHamiltonian} there now exists a zero-energy pure state $\ket{\chi}_{L R C} = \ket{\chi}_{AB}$ with 
    \begin{align}
        \|\chi_A - \ketbrab{\phi_L} \|_1 &\le \|\chi_B - \Phi(\psi)_L \|_1 + \|\Phi(\psi)_L - \ketbrab{\phi_L} \|_1 \\
        &\le \|\chi_{LR} - \Phi(\psi)_{LR}\|_1 + \|\Phi(\psi)_{LR} - \ketbrab{\phi_L} \otimes \ketbrab{\phi_R} \|_1 \\
        &\le \frac{2T}{T + L + 1} + \dy \\
        &\le O(L^{-1}).
    \end{align}
    Here we used the fact that trace distance is non-increasing under the partial trace. We could absorb $\dy$ into the $O(L^{-1})$ as $\dy$ is exponentially small and the parameter $L$ will be polynomially large.

    Note that $\ket{\phi_L}$ is pure and hence has zero entropy. Using the Fannes Inequality we thus have:
    \begin{align}
        S(\chi_L) \le O\left( \frac{1}{L}(n + \log L)\right).
    \end{align} 
    We now set $\El = 0$ and $t = O\left( \frac{1}{L}(n_R + \log L)\right)$ so that $\ket{\chi}_{ABC}$ satisfies $\braketb{\chi}{H_U} \le \El$ and $S(\chi_B) \le t$ as desired. Note that we will later set $L$ sufficiently large such that $t = O(1)$.

    For soundness, we first show that any pure state $\ket{\psi}_{LR}$ that is $\dn$-far from product has high entanglement entropy. For this, consider the Schmidt decomposition $\ket{\psi}_{LR} = \sum_i \lambda_i \ket{u_i}_L\ket{v_i}_R$. As $\psiab$ is $\dn$-far from product, it is in particular $\dn$-far from $\ket{u_i}\ket{v_i}$ for all $i$ hence we have: 
    \begin{align}
        \frac{1}{2} \|\psi - u_i \otimes v_i\|_1 &= \sqrt{1 - |\bra{u_i}\braket{v_i}{\psi}|^2 } \ge \frac{\dn}{2}\\
        |\lambda_i|^2 &= |\bra{u_i}\braket{v_i}{\psi}|^2 \le 1 - \frac{\dn^2}{4}
    \end{align}
    Recall the von Neumann entropy is lower bounded by the min-entropy. Hence we have:
    \begin{equation}
        \label{eq:entropylowerboundfromdistance}
        S(\psi_L) \ge H_\infty(\psi_L) = - \log \|\psi_L \|_\infty \ge - \log \sqrt{1 - \frac{\dn^2}{4}}.
    \end{equation}
    Setting $\dn = 2 - x$ yields $S(\psi_L) \ge - \log \sqrt{x - x^2/4} \ge - \frac{1}{2} \log x$. Recall now that $\PPIO$ is $\QMAt$-hard even when $\dn = 2 - 2^{-\Omega(n)}$. We conclude that in the NO case, all possible isometry outputs must have entropy at least $\Omega(n)$.

    Suppose now that for all $\ket{\psi}$, $\Phi(\psi)$ is $\dn$ far from any product state. We will first consider pure states: let $\ket{\psi}_{LRC}$ be such that $\braketb{\psi}{H} \le b'$. By the second property in \cref{Lem:ChanneltoHamiltonian} there exists a pure state $\ket{\chi}_L$ such that 
    \begin{align}
        D \coloneq \| \psi_{LR} - \Phi(\chi)\|_1 \le 2 \sqrt{\frac{b'}{\Delta}} + \frac{2T}{T + L + 1}.
    \end{align}
    As $\Phi(\chi)$ is far from a product state, we have $S(\Phi(\chi)_L) \ge \Omega(n)$ by the preceding discussion. By the Fannes Inequality we thus have 
    \begin{align}
        S(\psi_L) &\ge S(\Phi(\chi)_L) - D(n - \log D) \\
        &\ge \Omega(n) - \left(2 \sqrt{\frac{b'}{\Delta}} + \frac{2T}{T + L + 1}\right)\left(n - \log\left(2 \sqrt{\frac{b'}{\Delta}} + \frac{2T}{T + L + 1}\right) \right).
    \end{align}
    Setting $b' = \Delta/L \ge \Omega(L^{-4})$, and choosing the parameter $L$ to be a sufficiently large polynomial in $n$ we get $S(\psi_L) \ge \Omega(n)$ as desired.
    
    We now turn to mixed states. We have already shown that all low energy pure states have high entropy. Let $\rho_{LRC}$ be such that $\Tr(H\rho_{LRC}) \le b = \frac{b'}{n}$. We can write $\rho_{LRC}$ as a convex combination of pure states. All pure states with energy $\le b'$ will have entropy $\Omega(n)$, and the weight on pure states with energy $> b'$ cannot be more than $\frac{1}{n}$ without breaking the energy constraint. By concavity of the entropy, we then have $S(\rho_L) \ge \frac{n-1}{n} \Omega(n) = \Omega(n)$, completing the proof.    

    Close inspection of \cite{GHMW13} shows that a $\QMAt$-hard $\PPIO$ instance can be assumed to map $n$ qubits to $\Theta(n)$ qubits. It follows that the size of $B$ is polynomial in the size of $A$.
    
\end{proof}


\section{A $\QMAt$-complete \emph{local} Hamiltonian problem}\label{scn:leaps}
Recall the definition of $\LEAPS$.
\defleaps*

We will show that depending on the choice of parameters $\El, b, \dn, \dy$, $\LEAPS$ ``leaps'' from being $\QMA$-complete to being $\QMAt$-complete as claimed in \cref{thm:leapsQMA,thm:leapsQMA2}. This section will deal with general Hamiltonians. The treatment of physically motivated Hamiltonians is deferred to \cref{cor:physicalHamiltonians}

We begin by showing the $\QMA$ containment of $\LEAPS$ when $\dy$ is sufficiently smaller than the gap between $\El$ and $b$.

\begin{lemma}
    \label{lem:LEAPSQMAcontained}
    Let $D$ be an efficiently computable upper bound on $\|H\|_\infty$, such as the number of local terms $m$. Then $\LEAPS$ is contained in $\QMA$ when $b - \El - \dy D \ge \frac{1}{\poly(n)}$ and $\dn>\dy$. Note that no gap between $\dy$ and $\dn$ is required.
\end{lemma}
\begin{proof}
    We begin by noting that if $\dy = 0$, then $\LEAPS$ is just a special case of the Separable Local Hamiltonian Problem from \cite{CS12}, which is proven to be $\QMA$-complete in the same work. If $\dy > 0$, the H\"older Inequality yields
    \begin{equation}
        |\Tr(H\rho) - \Tr(H\sigma)| \le \|H(\rho - \sigma)\|_1 \le \|H\|_\infty \|\rho - \sigma\|_1 \le D \|\rho - \sigma\|_1.
    \end{equation} 
    Hence, if $\braketb{\psi}{H} \le \El$ and $\|\psi - \phi_A\otimes \phi_A\|_1 \le \dy$, then $\bra{\phi_B}\braketb{\phi_A}{H}\ket{\phi_B} \le \El + \dy D$. To solve the $\LEAPS$ instance $(H,\El, b, \dy, \dn)$ we can now solve the $\cfont{Separable Local Hamiltonian}$ instance $(H, \El', b)$ where $\El' = \El + \dy D \ge \El + \dy \|H\|_\infty$. Note that $b - \El' \ge \frac{1}{\poly(n)}$ by assumption. Furthermore, note that if $(H,\El, b, \dy , \dn)$ is a YES instance of $\LEAPS$, then $(H,\El',b)$ is a YES instance of $\cfont{Separable Local Hamiltonian}$. If $(H,\El, b, \dy , \dn)$ is a NO instance, then certainly all product states have energy $\ge b$, hence $(H, \El', b)$ is a NO instance as well. Note that we do not need a gap between $\dy$ and $\dn$ as this is subsumed by the energy gap.
\end{proof}

We now show the $\QMA$-hardness of $\LEAPS$, which follows almost immediately from \cite{CS12}.
\begin{lemma}
    \label{lem:LEAPSQMAhard}
    $\LEAPS$ is $\QMA$-hard, even when $\dy=0$.
\end{lemma}
\begin{proof}
    We reduce from the Separable Local Hamiltonian problem. Map an instance $(H, \El, b)$ to the $\LEAPS$ instance $(H, \El, b', 0, \dn)$, where $\dn = \frac{b - \El}{2D} \le \frac{b - \El}{2 \|H\|_\infty}$ and $b' = b - \dn D \le b - \|H\|_\infty \dn$. Note that $b' - \El \ge \frac{b - \El}{2} \ge \frac{1}{\poly(n)}$. Completeness follows trivially. For soundness, let $\ket{\psi}$ be a state that is $\dn$ close to a product state $\ket{\phi_L}\ket{\phi_R}$. By the H\"older inequality, $\braketb{\psi}{H} \ge \bra{\phi_B}\braketb{\phi_A}{H}\ket{\phi_B} - \|H\|_\infty \dn \ge b - \dn D = b'$, completing the proof.
\end{proof}

We now turn to showing $\QMAt$-completeness for different parameter regimes. We begin by showing containment in $\QMAt$, which holds as long as $b- \El$ and $\dn - \dy$ are at least inverse polynomial.

\begin{lemma}
    \label{lem:LEAPSQMA2contained}
    If $b - \El \ge \frac{1}{\poly(n)}$ and $\dn - \dy \ge \frac{1}{\poly(n)}$, then $\LEAPS$ is contained in $\QMAt$.
\end{lemma}
\begin{proof}
    We begin with the characterization $\QMA(k) = \QMAt$ from \cite{HM13}. The proofs will be $\ket{\psi_{AB}}, \ket{\phi_A}$ and $\ket{\phi_B}$. With probability $1/2$ the verifier performs the \emph{energy check} defined below. Otherwise they perform the \emph{distance check}. They accept iff the performed check accepts.
    \begin{enumerate}
        \item \emph{Energy check}: the verifier measures the energy\footnote{Technically, the energy of $H$ cannot be measured in one go. Here we implicitly use parallel repetition in the standard way.} $\braketb{\psi_{AB}}{H}$ and accepts if this is $< \frac{\El + b}{2}$.
        \item \emph{Distance check}: the verifier uses the SWAP test between $\ket{\psi_{AB}}$ and $\ket{\phi_A}\ket{\phi_B}$ and accepts if the SWAP test accepts.
    \end{enumerate}
    For completeness, the provers send $\ket{\psi_{AB}}, \ket{\phi_A}, \ket{\phi_B}$ such that $\braketb{\psi_{AB}}{H} \le \El$ and $\| \psi_{AB} - \phi_A \otimes \phi_B\|_1 \le \dy$. The energy check can be made to succeed with probability $1 - 2^{-\Omega(n)}$ using parallel repetition. The distance check will accept with probability
    \begin{align}
        p_{acc} = \frac{1}{2} + \frac{1}{2}\left| \braket{\psi_{AB}}{\phi_A}\ket{\phi_B} \right|^2
        = 1 - \frac{1}{8} \| \psi_{AB} - \phi_A \otimes \phi_B\|_1^2 \ge 1 - \frac{\dy^2}{8}.
    \end{align}
    The total acceptance probability is therefore at least $1 - 2^{-\Omega(n)} - \frac{\dy^2}{16} \ge \frac{3}{4} + \frac{1}{\poly(n)}$, where we use the fact that $\dy \le \dn - \frac{1}{\poly(n)} \le 2 - \frac{1}{\poly(n)}$.

    For soundness we can assume the given proofs to be pure by a standard convexity argument. Note that if $\braketb{\psi_{AB}}{H} \ge b$, then the energy check will fail with probability inverse exponentially close to 1. Hence, if the provers send such $\ket{\psi_{AB}}$, the verifier accepts with probability at most $\frac{1}{2} + 2^{-\Omega(n)}$, which is sufficiently far from the acceptance probability in the YES case to distinguish both cases. We can thus assume $\braketb{\psi_{AB}}{H} < b$. As we are in the NO case, this implies $\| \psi_{AB} - \phi_A\otimes\phi_B\|_1 > \dn$. The acceptance probability of the distance check will now be at most $1 - \frac{\dn^2}{8}$, which means the total acceptance probability is $\le 1 - \frac{\dn^2}{16}$, which is inverse polynomially bounded away from $1 - 2^{-\Omega(n)} - \frac{\dy^2}{16}$, the acceptance probability in the YES case. This completes the soundness proof.
\end{proof}

Finally, we show that if the distance thresholds $\dy,\dn$ are sufficiently larger than the energy thresholds $\El, b$, then $\LEAPS$ is $\QMAt$-hard.
\begin{lemma}
    \label{lem:LEAPSQMA2hard}
    If $\dy = \frac{1}{\poly(n)}$ and $b = O(\dy^6)$ (so in particular, $b < \dy$), then $\LEAPS$ is $\QMAt$-hard.
\end{lemma}

Note that the regime where we show $\QMAt$-hardness has $b < \dy$ but our $\QMA$-containment result requires $\dy < b$ (as $\|H\|_\infty \ge 1$). There is hence no overlap between the parameter regimes for which we show $\QMA$ and $\QMAt$ completeness. It is in theory possible, although perhaps unlikely, that the $b$ and $\dy$ parameters could be tweaked such that there is overlap between the $\QMAt$-hard and contained in $\QMA$ regime. This would prove $\QMA = \QMAt$.

\begin{proof}
    We will reduce from the $\QMA(2)$-complete problem $\PPIO$ (\Cref{def:PPIO}). Let $(U, \dy', \dn')$ be a $\PPIO$ instance where $U$ is a unitary acting on registers $L, R$ of $n_L + n_R = n$ qubits. Assume that $\dy' = 2^{-\Omega(n)}$ and $\dn' = 2 - 2^{-\Omega(n)}$. We again use the channel-to-Hamiltonian construction from \cref{Lem:ChanneltoHamiltonian}, this time applied to the channel $\Phi(\rho) \coloneq \Tr_L\left( U(\rho_L \otimes \ketbrab{0^{n_R}}) U^\dagger \right)$. Let $H_U$ be the resulting Hamiltonian acting on registers $L,R,C$. We set $A \coloneq LC$ and $B \coloneq R$. The output of the reduction will be $(H_U, \El, b, \dy, \dn)$ with $\El = 0, b = \Omega(L^{-3}), \dy = O(L^{-1/2})$ and $\dn = O(1)$. Note that $H_U$ is indeed $5$-local.

    For completeness, let $\ket{\psi}, \ket{\phi_L}, \ket{\phi_R}$ be such that $\|U(\psi \otimes \ketbrab{0^{n_B}})U^\dagger - \phi_L \otimes \phi_R\| \le \dy'$. Trace distance is non-increasing under partial trace, so in particular $\|\Phi(\psi) - \phi_R\|_1 \le \dy'$. By \cref{Lem:ChanneltoHamiltonian} property 1 there now exists a zero-energy pure state $\ket{\chi}$ such that $\|\chi_B - \phi_R\|_1 \le \frac{2T}{T + L + 1} + \dy'$. 
    
    We now prove that if the reduced state is close to pure, the global state is close to product:
    \begin{lemma}
        \label{lem:7.5}
        Let $\ket{\psi}_{XY}$ be a pure state such that $\|\psi_Y - \ketbrab{\phi}_Y\|_1 \le \delta$ for some pure state $\ket{\phi}_Y$. Then there exists a state $|\tilde{\phi}\rangle_X$ with
        \begin{align}
            \|\ketbrab{\psi}_{XY} - \ketbrab{\tilde{\phi}}_X \otimes \ketbrab{\phi}_Y\|_1 \le \sqrt{2\delta}.
        \end{align}
    \end{lemma}
    \begin{subproof}[Proof of \cref{lem:7.5}]
        By the Fuchs-van de Graaf inequality (using that $\ket{\phi}$ is pure): $F(\psi_Y, \ketbrab{\phi}_Y) \ge 1 - \frac{\delta}{2}$. By Uhlmann's theorem, this fidelity is equal to the largest overlap of $\ket{\psi}_{XY}$ with a purification of $\ket{\phi}_Y$. As $\ket{\phi}_Y$ is already pure, we thus have: 
        \begin{align}
            F(\psi_Y, \ketbrab{\phi}_Y) = \max_{\ket{\chi}_X}|\braket{\psi}{\chi}\ket{\phi}|^2.
        \end{align}
        Letting $|\tilde{\phi}\rangle$ be the state maximizing this overlap, and using the formula for trace distance between pure states, we get:
        \begin{align}
            \|\ketbrab{\psi}_{XY} - \ketbrab{\tilde{\phi}}_X \otimes \ketbrab{\phi}_Y\|_1 &= 2\sqrt{1 - |\langle\psi|\tilde{\phi}\rangle\ket{\phi}|^2} \\
            &\le \sqrt{2\delta}
        \end{align}
    \end{subproof}    
    
    Applying this lemma to $\chi_{AB}$ now immediately yields the existence of a state $|\tilde{\phi}\rangle_{A}$ such that
    \begin{align}
        \|\chi_{AB} - \ketbrab{\tilde{\phi}}_{A} \otimes \ketbrab{\phi_{R}}_{B} \| &\le \sqrt{2}\sqrt{\frac{2T}{T + L + 1} + \dy'} \\ 
        &\le O\left(\frac{1}{\sqrt{L}}\right),
    \end{align}
    where in the last step we use that $L$ is polynomial in $n$ to absorb the $\dy' = 2^{-\Omega(n)}$. Setting $\dy = \Omega\left(\frac{1}{\sqrt{L}}\right)$ we see that $\ket{\chi}$ has zero energy and is $\dy$-close to a product state as desired.

    For soundness, assume $(U, \dy', \dn')$ is a NO instance and let $\ket{\psi}_{AB}$ be an arbitrary state with $\braketb{\psi}{H} \le b$. By \cref{Lem:ChanneltoHamiltonian} property 2, there exists a pure state $\ket{\chi}_L$ such that $\|\psi_B - \Phi(\chi)\|_1 \le \sqrt{b / \Delta} + \frac{2T}{T + L + 1}$. Suppose now, towards a contradiction, that $\ket{\psi}$ is close to product: suppose there are $\ket{\phi_{A}}, \ket{\phi_B}$ such that $\|\psi_{AB} - \phi_{A}\otimes \phi_{B}\|_1 \le \dn$. By monotonicity under partial trace of trace distance we have:
    \begin{align}
        \|\Phi(\chi) - \phi_B\|_1 &\le \|\Phi(\chi) - \psi_B\|_1 + \|\psi_B - \phi_B\|_1 \\
        &\le \|\Phi(\chi) - \psi_B\|_1 + \|\psi_{AB} - \phi_{A}\otimes \phi_B\|_1 \\
        &\le \sqrt{b / \Delta} + \frac{2T}{T + L + 1} + \dn.
    \end{align}

    Note that $U\ket{\chi} \otimes \ket{0^{n_R}}$ is a pure state and that $\Phi(\chi)$ is its reduced state on the $R$ register. \cref{lem:7.5} thus yields the existence of a product state $|\tilde{\phi}_{L}\rangle \otimes \ket{\phi}_R$ with 
    \begin{align}
        \| U\ketbrab{\chi}_L \otimes \ketbrab{0^{n_R}}_R U^\dagger - |\tilde{\phi}\rangle\langle\tilde{\phi}|_{L}\otimes \ketbrab{\phi}_R\| \le \sqrt{2}\sqrt{\sqrt{b / \Delta} + \frac{2T}{T + L + 1} + \dn}.
    \end{align}
    But as we are in the NO case, $U\ket{\chi} \otimes \ket{0^{n_R}}$ is $\dn'$ far from any pure product state. 
    We have thus achieved a contradiction if 
    \begin{align}
        \sqrt{2}\sqrt{\sqrt{b / \Delta} + \frac{2T}{T + L + 1} + \dn} < \dn'.
    \end{align}
    In particular, setting $b \le \Delta$, $\dn = O(1)$ and $L = n^{\Omega(1)}$ achieves $\sqrt{b / \Delta} + \frac{2T}{T + L + 1} + \dn \le O(1)$ and hence suffices. Note that $\Delta = \Omega(L^{-3})$, hence setting $b = O(L^{-3})$ suffices. We now have $\dy = \Omega\left(\frac{1}{\sqrt{L}}\right)$ and $b = O\left(\frac{1}{L^3}\right)$, hence $b = O(\dy^6)$.
\end{proof}


\section{Physical Hamiltonians}
\label{sec:physicalHamiltonians}
In this section we will use the universal Hamiltonian simulator framework developed in \cite{CMP18,ZA21} to extend our results to physically motivated Hamiltonians. The main idea of this framework is to encode any Hamiltonian $H$ into the low energy space of a bigger Hamiltonian $H'$, which belongs to some set of physically relevant Hamiltonians of choice, such as Heisenberg Hamiltonians or $XY$-interaction Hamiltonians.
We use a simplified notion of Hamiltonian simulation than that originally introduced in \cite{CMP18}. 

\begin{definition}[Hamiltonian simulation (simplified version of {\cite[Definition~23]{CMP18}})]
    \label{def:HamiltonianSimulation}
    An $n$-qudit Hamiltonian $H$ is $(\Delta, \eta, \epsilon)$-simulated by $H'$ if for some local isometry $V = \bigotimes_i V_i$, where each $V_i$ maps a single qubit to $O(1)$ qubits, there is a (not necessarily local) isometry $\tilV$ such that
    \begin{enumerate}
        \item $V$ is close to $\tilV$: $\|\tilV - V\|_{\infty} \le \eta$.
        \item $\tilV$ maps the identity to the projector on the subspace of states with energy $\le \Delta$ according to $H'$. That is, $\tilV I \tilV^\dagger = P_{\le \Delta(H')}$ 
        \item The low energy part of the spectrum of $H'$ is close to the spectrum of $H$. That is, $\|H'_{\le \Delta} - \tilV H\tilV^\dagger\|_{\infty} \le \epsilon$, where $H'_{\le \Delta} = H'P_{\le \Delta(H')}$.
    \end{enumerate}
    We say the simulation is \emph{efficient} if the description of $H'$ can be efficiently computed and both number of particles $H'$ acts on, and its maximum energy $\|H'\|_{\infty}$ are polynomial in $n, \eta^{-1}, \epsilon^{-1}, \Delta$.
\end{definition}

\begin{definition}[Universal family of Hamiltonians (\cite{CMP18, ZA21})]
    A family $\calF = \{H_i\}$ is \emph{weakly universal} if for any $\Delta, \eta, \epsilon > 0$, any $O(1)$-local Hamiltonian $H$ can be $(\Delta, \eta, \epsilon)$-simulated by some $H_i \in \calF$. We say the family is \emph{strongly universal} if this simulation is always efficient.
\end{definition}

The main result of \cite{ZA21} says that most families of Hamiltonians on the 2D square lattice are strongly universal. The exception are families that are 2SLD, meaning that their 2-local interactions can be simultaneously and locally diagonalized. More precisely:
\begin{definition}[2SLD (\cite{CMP18})]
    Let $\calS$ be a set of 2-local interaction terms on qubits. $\calS$ is \emph{2SLD} if there exists $U \in \SUt$ such that $U^{\otimes 2} H_i (U^\dagger)^{\otimes 2} = \El_i Z \otimes Z + A_i \otimes I + I \otimes B_i$ for all $H_i \in \calS$. Here $\El_i \in \RR$ and $A_i, B_i$ are arbitrary 1-local operators.
\end{definition}

Originally, \cite{ZA21} uses the fully general definition of Hamiltonian simulation. However, closer inspection of their results reveals that their simulations have the form of \cref{def:HamiltonianSimulation}.

\begin{theorem}[Main theorem of \cite{ZA21}]
    \label{Thm:StronglyUniversal}
    Let $\calS$ be a set of 2-qubit interactions and $\calF(\calS)$ be the family of Hamiltonians on a 2D square lattice where all 2-qubit interactions are from $\calS$, with arbitrary interaction strengths. If $\calS$ is \emph{non-}2SLD, then $\calF(\calS)$ is strongly universal.
\end{theorem}
\begin{corollary}
    The families of Heisenberg Hamiltonians ($\calS = \{X \otimes X + Y \otimes Y + Z \otimes Z\}$) and $XY$-interaction Hamiltonians ($\calS = \{X \otimes X + Y \otimes Y\}$) are strongly universal.
\end{corollary}

Using this framework we easily get that all hardness results in this work hold even when restricted to any non-2SLD family of Hamiltonians.

\begin{corollary}
    \label{cor:physicalHamiltonians}
    Let $\calS$ be a non-2SLD set of $2$-qubit interactions and let $\calF(\calS)$ be as in \cref{Thm:StronglyUniversal}. Then
    \begin{enumerate}
        \item $\HELES$ is $\qqQAM$-hard, even when restricted to $\calF(\calS)$.
        \item $\LELES$ is $\QMAt$-hard, even when restricted to $\calF(\calS)$
        \item $\LEAPS$ is $\QMAt$-hard, even when restricted to $\calF(\calS)$
    \end{enumerate}
\end{corollary}
\begin{proof}
    We will first consider $\HELES$. The proof for $\LELES$ will be completely analogous.

    Let $(H_{AB}, \El, b, s ,t)$ be a $\HELES$ instance and let $H'_{A'B'}$ $(\Delta, \eta, \epsilon)$-simulate $H$ via the local isometry $V$. We map our instance to $(H'_{A'B'}, \El', b', s', t')$ where $\El' = \El + \frac{b - \El}{4}, b' = b - \frac{b - \El}{4}, s' = s$ and $t' = t + \frac{s - t}{4}$. We now show that YES and NO instances are preserved.
    
    For completeness, let $\ket{\psi}_{AB}$ be a state with $\braketb{\psi}{H} \le \El$. Consider the state $\ket{\psi'} = V\ket{\psi}$. 
    We have:
    \begin{align}
        \braketb{\psi'}{H'} &= \braketb{\psi}{V^\dagger H'V}\\
        &\le \braketb{\psi}{\tilV^\dagger H' \tilV} + \|H'\|_\infty \|\tilV - V\|_\infty \\
        &\le \braketb{\psi}{H} + \epsilon + \eta\|H'\|_\infty \\
        &\le \El + \epsilon + \eta\|H'\|_\infty.
    \end{align}
    Choosing $\eta \le \frac{b - \El}{8 \|H'\|_\infty}$ and $\epsilon \le \frac{b - \El}{8}$ we get that $\braketb{\psi'}{H'} \le \El'$ is low energy. Furthermore, as $V$ is a local isometry it preserves entropy. In particular, $S(\psi_A) = S(\psi'_{A'})$.

    For soundness suppose $\rho'$ satisfies $\Tr(H'\rho) \le b'$. Let $\sigma' = \frac{\Pi_{\le \Delta} \rho' \Pi_{\le \Delta}}{\Tr(\Pi_{\le \Delta} \rho')}$ be the projection of $\rho'$ on the low energy space of $H'$. Note that 
    \begin{align}
        b' &\ge \Tr(H'\rho') \\
        &\ge \Delta\Tr\left((I - \Pi_{\le \Delta}) \rho' \right) \\
        &\ge \Delta\left(1 - \Tr\left(\Pi_{\le \Delta} \rho'\right)\right) \\
        \Tr(\Pi_{\le \Delta} \rho') & \ge 1 - \frac{b'}{\Delta},
    \end{align}
    hence by the Gentle Measurement Lemma, $\|\rho' - \sigma'\|_1 \le 2 \sqrt{\frac{b'}{\Delta}}$. 

    As $\sigma' \in \text{Im}(\tilV)$, there is some $\sigma$ such that $\sigma' = \tilV\sigma \tilV^\dagger$. We now have
    \begin{align}
        \Tr(H \sigma) &\le \Tr(\sigma'\tilV H \tilV^\dagger) \\
        &\le \Tr(H'\sigma') + \epsilon \\
        &\le \Tr(H'\rho') + 2\|H'\|_\infty \sqrt{\frac{b'}{\Delta}} + \epsilon.
    \end{align}
    Hence, for appropriate choices of $\Delta, \epsilon$, $\Tr(H\sigma) \le b$. As we are in the NO case, this means that $S(\sigma_A) \le t$ and as $V$ is a local isometry, $S(\chi_A) \le t$ for $\chi = V\sigma V^\dagger$. Note that $\|\chi - \sigma'\|_1 \le \|V - \tilde{V}\|_\infty = \eta$ and that $\|\sigma' - \rho'\|_1 \le 2\sqrt{\frac{b'}{\Delta}}$, hence $\|\chi - \rho'\|_1 \le \eta + 2 \sqrt{\frac{b}{\Delta}}$. By the Fannes Inequality, the entropy of $\rho'_{A'}$ is now close to that of $\chi_A$. In particular, for appropriate choices of $\eta, \Delta$, we get $S(\rho'_{A'}) \le t'$.

    The proof for $\LELES$ follows by swapping the roles of $t$ and $s$ and the corresponding inequalities. 

    For $\LEAPS$, the situation is a bit more subtle as, at least in principle, there could be product states that are only product because of a very small contribution of the high energy space of $H'$. Nevertheless, the proof still proceeds largely along the same lines as $\HELES$. We map our $\LEAPS$ instance $(H, \El, b, \dy, \dn)$ to $(H', \El', b', \dy', \dn')$ where $H', \El', b'$ are as above. For the $\dy, \dn$ parameters we will set $\dy' = \dy = O(n^{-3})$, $\dn = \Omega(1)$ and $\dn' = \Theta(n^{-2})$. Note that the $\QMAt$-hardness proof of $\LEAPS$ indeed allows these choices of $\dy,\dn$.

    For completeness it suffices to note that $V\ket{\phi}_{AB}$ is at most as from product as $\ket{\phi}_{AB}$ since $V$ is a local isometry.

    For soundness, we need to show that if $\ket{\phi}_{AB}$ is far from product, then so is $\ket{\psi'}_{A'B'}$. Towards a contradiction, assume there exist $\ket{\theta}_{A'}, \ket{\theta}_{B'}$ such that $\|\psi'_{A'B'} - \theta_{A'} \otimes \theta_{B'}\|_1 \le \dn'$. It follows that $\|\phi'_{A'B'} - \theta_{A'} \otimes \theta_{B'}\|_1 \le \dn' + 2 \sqrt{\frac{b}{\Delta}} \eqcolon D$. By the Fannes Inequality, we now have $S(\phi'_{A'}) \le nD - D\log D$. As $V$ is a local isometry it preserves entropy, hence $S(\phi_A) \le nD - D\log D$ as well. 

    Just like before we have $\braketb{\phi}{H} \le b$ by our choice of parameters. As we are in the NO case, this implies that $\ket{\phi}_{AB}$ is $\dn$-far from any product state. In particular, $\ket{\phi}_{AB}$ is $\dn$-far from the product states in its own Schmidt decomposition, and $S(\phi_A) \ge - \log \sqrt{1 - \frac{\dn^2}{4}}$ (similarly to \cref{eq:entropylowerboundfromdistance}). As we have $\dn = \Omega(1)$, we get our desired contradiction if $nD - D \log D = o(1)$. The proof now follows by setting $\Delta = n^4b$ so that $D = O(n^{-2})$.
    
\end{proof}

\section*{Acknowledgements.} The authors would like to thank Dorian Rudolph for bringing the class $\qqQAM$ to our attention.

\printbibliography 

\end{document}